\documentclass[sigplan,10pt]{acmart}
\acmConference[PL'18]{ACM SIGPLAN Conference on Programming Languages}{January 01--03, 2018}{New York, NY, USA}
\acmYear{2018}
\acmISBN{} 
\acmDOI{} 
\startPage{1}

\setcopyright{none}

\makeatletter                   
\def\mdseries@tt{m}             
\makeatother                    

\usepackage{booktabs}        
\usepackage{subcaption}      

\usepackage{threeparttable}  
\usepackage{graphicx}        
\usepackage{semantic}        
\usepackage{algorithm}       
\usepackage{algpseudocode}   
\usepackage{cleveref}        
\usepackage{xcolor}          
\usepackage{multicol}

\usepackage{microtype}
\usepackage[subtle]{savetrees}

\newcommand{\Projname}{\textsc{Castor}}
\newcommand{\projname}{\textsc{Castor}}
\newcommand{\langname}{layout algebra}

\newcommand{\papertitle}{Deductive Optimization of Relational Data Storage}
\newcommand{\demomatch}{\textsc{Demo\-Match}}
\newcommand{\hyper}{\textsc{Hyper}}
\newcommand{\cozy}{\textsc{Cozy}}
\newcommand{\chestnut}{\textsc{Chestnut}}

\newcommand{\concat}{\ensuremath{\mathbin{+\!\!+}}}

\newcommand{\tuple}[2]{\textsf{\textbf{tuple}}_{#2}(#1)}
\newcommand{\ctuple}[1]{\tuple{#1}{\textsf{cross}}}

\newcommand{\static}[1]{\boxed{#1}}
\newcommand{\partition}[3]{\textsc{part} (#1, #2, #3)}

\newcommand{\cxform}{\xrightarrow{t}}

\newcommand{\todo}[1]{}
\renewcommand{\todo}[1]{{\color{red} TODO: {#1}}}

\newcommand{\as}[2]{#1\ \textsf{as}\ #2}
\newcommand{\filter}{\textsf{filter}}
\newcommand{\join}{\textsf{join}}
\newcommand{\depjoin}{\textsf{depjoin}}
\newcommand{\select}{\textsf{select}}
\newcommand{\dedup}{\textsf{dedup}}
\newcommand{\groupby}{\textsf{group-by}}
\newcommand{\orderby}{\textsf{ord\-er-by}}
\newcommand{\scalar}[1]{\textsf{\textbf{scalar}}(#1)}
\newcommand{\nscalar}{\textsf{\textbf{scalar}}}
\newcommand{\alist}{\textsf{\textbf{list}}}
\newcommand{\atuple}[1]{\textsf{\textbf{tuple}}_{#1}}
\newcommand{\ntuple}{\textsf{\textbf{tuple}}}
\newcommand{\hidx}{\textsf{\textbf{hash-idx}}}
\newcommand{\oidx}{\textsf{\textbf{order\-ed-idx}}}
\newcommand{\ite}[3]{\textsf{if}\ #1\ \textsf{then}\ #2\ \textsf{else}\ #3}

\newcommand{\layoutto}{\downarrow}

\newcommand{\ccol}[1]{\omit\hfill #1\hfill}

\newcommand{\stitle}[1]{\vspace{1ex}\noindent{\bf #1}}

\newcommand{\evalto}{\Downarrow}
\newcommand{\eto}[3][\sigma, \delta]{\ensuremath{#1 |- #2 \Downarrow #3}}

\Crefname{section}{Sec.}{Sections}

\makeatletter
\def\@copyrightspace{\relax}
\makeatother

\begin{document}

\title{\papertitle}         



\author{John K. Feser}
\orcid{0000-0001-8577-1784}             
\affiliation{
  \department{CSAIL}              
  \institution{MIT}            
  \streetaddress{32 Vassar St.}
  \city{Cambridge}
  \state{MA}
  \postcode{02139}
  \country{USA}                    
}
\email{feser@csail.mit.edu}          

\author{Samuel Madden}
\affiliation{
  \department{CSAIL}              
  \institution{MIT}            
  \streetaddress{32 Vassar St.}
  \city{Cambridge}
  \state{MA}
  \postcode{02139}
  \country{USA}                    
}
\email{madden@csail.mit.edu}          

\author{Nan Tang}
\affiliation{
  \department{Qatar Computing Research Institute}             
  \institution{QCRI HBKU}           
   \city{Doha}
  \country{Qatar}                   
}
\email{ntang@hbku.edu.qa}         

\author{Armando Solar-Lezama}
\affiliation{
  \department{CSAIL}              
  \institution{MIT}            
  \streetaddress{32 Vassar St.}
  \city{Cambridge}
  \state{MA}
  \postcode{02139}
  \country{USA}                    
}
\email{asolar@csail.mit.edu}          

\begin{abstract}
  Optimizing the physical data storage and retrieval of data are two key
  database management problems. In this paper, we propose a language that can
  express a wide range of physical database layouts, going well beyond the row-
  and column-based methods that are widely used in database management systems.
  We use deductive synthesis to turn a high-level relational representation of a
  database query into a highly optimized low-level implementation which operates
  on a specialized layout of the dataset. We build a compiler for this language
  and conduct experiments using a popular database benchmark, which shows that
  the performance of these specialized queries is competitive with a
  state-of-the-art in memory compiled database system.

\end{abstract}



\maketitle

\section{Introduction}

Traditional database systems are generic and powerful, but they are not well optimized for static databases.
A static database is one where the data changes slowly or not at all and the queries are fixed.
These two constraints introduce opportunities for aggressive optimization and specialization.
This paper introduces \projname{}: a domain specific language and compiler for building static databases.
\projname{} achieves high performance by combining query compilation techniques from state-of-the-art in-memory databases~\cite{Neumann2011} with a deductive synthesis approach for generating specialized data structures.

To better understand the scenarios that \projname{} supports, consider these two use cases.
First, consider a company which maintains a web dashboard for displaying internal analytics from data that is aggregated nightly.
The queries used to construct the dashboard cannot be precomputed directly, because they contain parameters like dates or customer IDs, but there are only a few query templates.
Not all the data in the original database is needed, and some attributes are only used in aggregates.
As another example, consider a company which is shipping a GPS device that contains an embedded map. The map data is infrequently updated, and the device queries it in only a few specific ways.
The GPS manufacturer cares more about compactness and efficiency than about generality.
As with the company building the dashboard, it is desirable to produce a system that is optimal for the particular dataset to be stored. 

These two companies could use a traditional database system, but using a system designed to support arbitrary queries will miss important optimization opportunities.
Alternatively, they could write a program using custom data structures.
This will give them tight control over their data layout and query implementation but will be difficult to develop and expensive to maintain.

\begin{figure}
  \includegraphics[width=0.7\columnwidth]{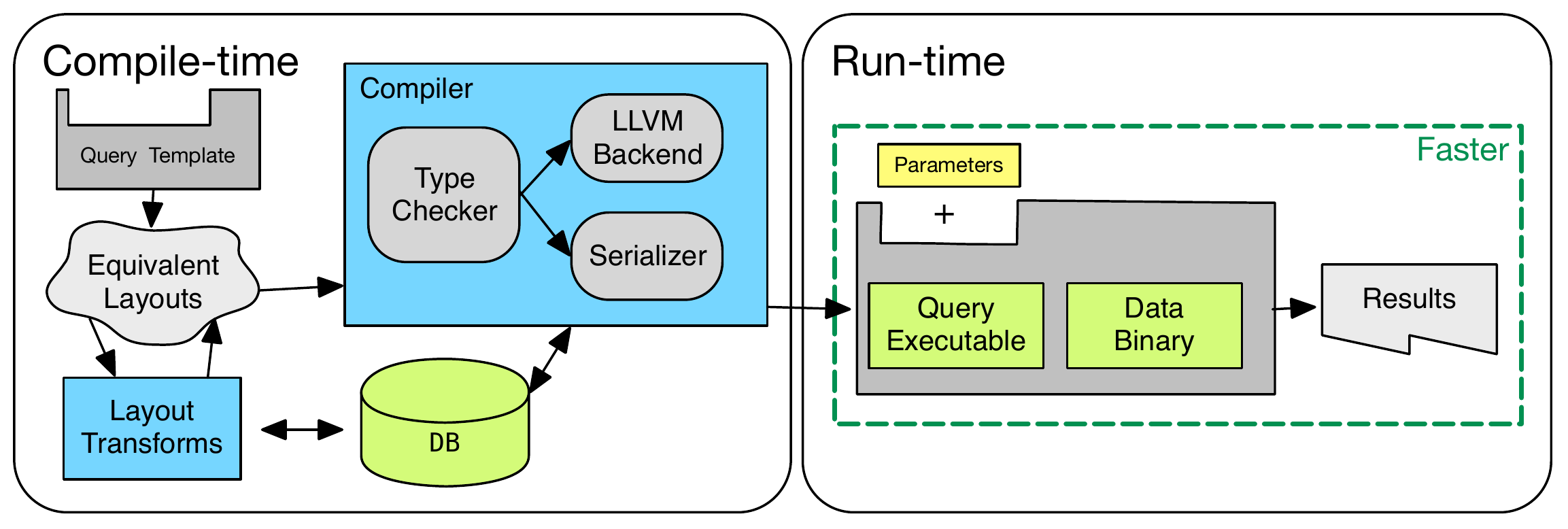}
  
  \includegraphics[width=0.6\columnwidth]{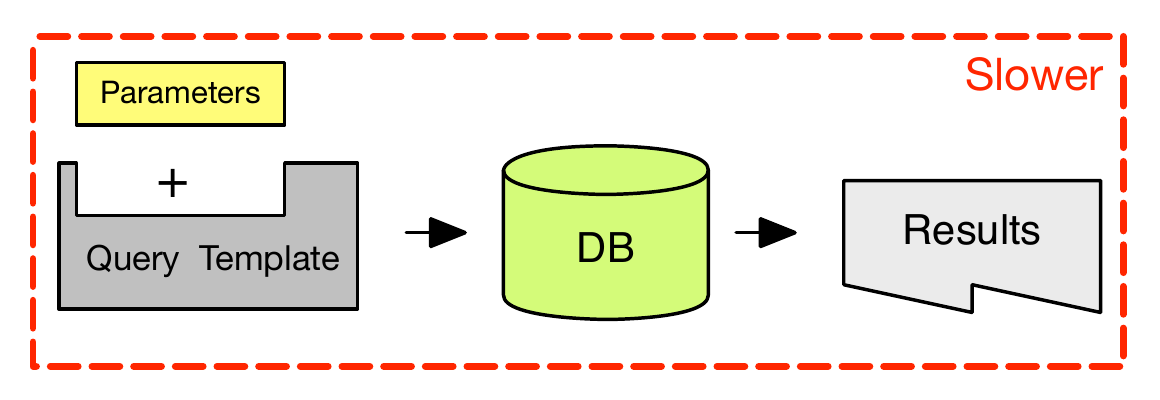}
  \vspace*{-2ex}
  \caption{An overview of the \projname{} system (above) vs a traditional RDBMS (below).}\label{fig:overview}
  \vspace*{-4ex}
\end{figure}



\projname{} is an attempt to capture some of the optimization opportunities
of static databases and to address the needs of these two scenarios. 
As \Cref{fig:overview} illustrates, the input to \projname{} is a dataset and a parameterized query 
that a client will want to invoke on the data. The user then interacts with \projname{} 
using high-level commands to generate an efficient implementation 
of an in-memory datastore specialized for the dataset and the parameterized query. The commands available in \projname{} give the programmer
tight control over the exact organization of the data in memory, allowing the user
to trade off memory usage against query performance without the risk of introducing
bugs.  \projname{} also uses code generation
techniques from high-performance in-memory databases to produce the low-level implementations
required for efficient execution. The result is a package of data and code
that uses significantly less memory than the most efficient in-memory databases and 
for some queries can even surpass the performance of in-memory databases that already rely on 
aggressive code generation and optimization~\cite{Neumann2011}.

\subsection{Contributions}

\projname{} is made possible by three major technical contributions: a new
notation to jointly represent the layout of the data in memory and the queries
that will be computed on it, a set of deductive optimization rules that
generalize traditional query optimization rules to jointly optimize the query
and the data layout, and a type-driven layout compiler to produce both a binary
representation of the data from the high-level data representation and
specialized machine code for accessing it. 



\vspace*{-1ex}
\paragraph*{Integrated Layout \& Query Language}
We define the \emph{layout algebra}, which extends the relational
algebra~\cite{Codd1970} with layout operators that describe the particular
data items to be stored and the layout of that data in memory. The layout
algebra is flexible and can express many layouts, including row stores
and clustered indexes. It supports nesting layouts, which gives control over
data locality and supports prejoining of data. Our use of a language which
combines query and layout operators makes it possible to write deductive
transformations that change both the runtime query behavior and the data layout.

\vspace*{-1ex}
\paragraph*{Deductive Optimization Rules}
\projname{} provides a set of equivalence preserving transformations which can
change both the query and the data layout. The user can apply these
transformations to deductively optimize their query without worrying about
introducing bugs. Alternatively, they can use \projname{}'s optimizer, which
automatically selects a sequence of transformations. \Projname{}'s specialized
notation turns transformations that would be complex and global in other
database systems into local syntactic changes.

\vspace*{-1ex}
\paragraph*{Type-driven Layout Compiler}
Existing relational synthesis tools use standard library data structures and
make extensive use of pointer based data structures that hurt
locality~\cite{Loncaric2018,Loncaric2016,Hawkins2011}. \projname{} uses a
specializing layout compiler that takes the properties of the data into account
when serializing it. Before generating the layout, \projname{} generates an
abstraction called a \emph{layout type} which guides the layout specialization.
For example, if the layout is a row-store with fixed-size tuples, the layout
compiler will not emit a length field for the tuples. Instead, this length will
be compiled directly into the query. This specialization process creates very
compact datasets and avoids expensive branches in generated code.

\vspace*{-1ex}
\paragraph*{High Performance Query Compiler}
\projname{} uses code generation techniques from the high performance in-memory
database literature~\cite{Neumann2011,Shaikhha2016,Tahboub2018,Rompf2015}. It
eschews the traditional iterator based query execution model~\cite{Graefe1994}
in favor of a code generation technique that produces simple, easily optimized
low-level code. \projname{} directly generates LLVM IR and augments the
generated IR with information from the layout type that allows LLVM to further
optimize it.

\vspace*{-1ex}
\paragraph*{Empirical Evaluation}
We empirically evaluate \projname{} on a benchmark derived from TPC-H, a
standard database benchmark~\cite{Council2008}. We show that \projname{} is
competitive with the state of the art in-memory compiled database system
\hyper{}~\cite{Neumann2011} while using significantly less memory. We also show
that \projname{} scales to larger queries than the leading data-structure
synthesis tool \cozy{}~\cite{Loncaric2018}.

\subsection{Limitations}

\emph{\Projname{} constructs read-only databases.} This design decision limits
the appropriate use cases for \projname{} but it enables important
optimizations. \Projname{} takes advantage of the absence of updates to tightly
pack data together, which improves locality. \Projname{} also aggressively
specializes the compiled query by including information about the layout, such
as lengths of arrays and offsets of layout structures. Providing this
information to the compiler improves the generated code.

\emph{Only one parameterized query can be optimized at once.} This is a
limitation, but it is not a serious one. A multi-query workload can be supported
by replicating the dataset and optimizing each query separately. \Projname{}
removes any data which is not needed by the query and it produces compact
layouts for the data that remains. This reduces the overhead of the replication.


\section{Motivating Example}\label{sec:example}

We now describe the operation of \projname{} on an application from the software
engineering literature. \demomatch{} is a tool which helps users understand
complex APIs using software demonstrations~\cite{Yessenov2017}.
\demomatch{} maintains a database of program traces---computed offline---which
it queries to discover how to use an API.\@
\demomatch{} is a good fit for \projname{}. The data in question is largely
static: computing new traces is an infrequent task. The data is automatically
queried by the tool, so there is no need to support ad-hoc queries.
Finally, query performance is important for \demomatch{} to work 
as an interactive tool. 

\subsection{Background}\label{sec:background}

\demomatch{} stores program traces as ordered collections of \emph{events}
(e.g., function calls). Traces have an inherent tree structure: each event has
an \texttt{enter} and an \texttt{exit} and nested events may occur between the
enter and exit.

\begin{figure}
  \vspace*{-2ex}
  \centering
  \begin{subfigure}[b]{0.12\textwidth}
    \centering
    \includegraphics[width=\textwidth]{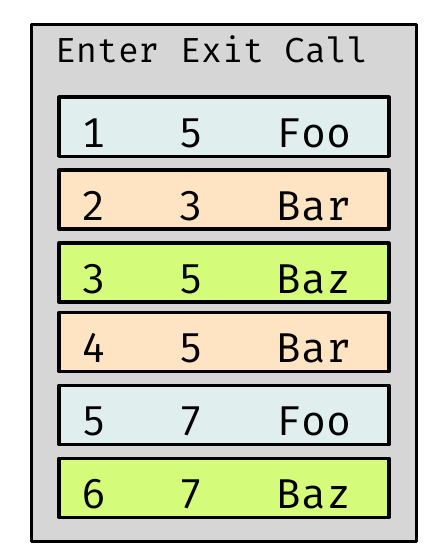}
    \caption{In tabular form.}
  \end{subfigure}\hspace{1em}%
  \begin{subfigure}[b]{0.12\textwidth}
    \centering
    \includegraphics[width=\textwidth]{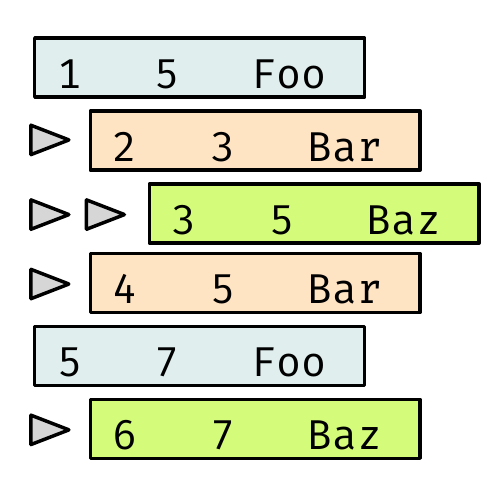}
    \caption{In tree form.}
  \end{subfigure}\hspace{1em}%
  \begin{subfigure}[b]{0.12\textwidth}
    \centering
    \includegraphics[width=\textwidth]{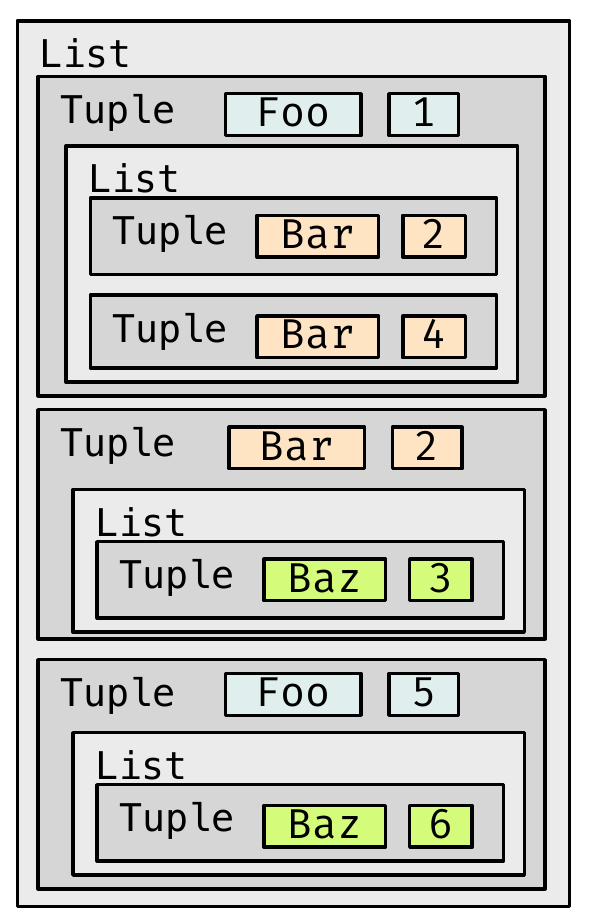}
    \caption{As a nested layout.}
  \end{subfigure}
  \vspace*{-2ex}
  \caption{Graphical representation of the \demomatch{} data.}\label{fig:demomatch}
  \vspace*{-1ex}
\end{figure}


A critical query in the \demomatch{} system finds nested calls to particular
functions in a trace of program events:
{\small
\begin{align*}
&\textsf{select}\ p.enter,\ c.enter\ \textsf{from}\ log\ \textsf{as}\ p,\ log\ \textsf{as}\ c\ \textsf{where} \\
&\quad p.enter < c.enter \land c.enter < c.exit \land p.id = \$pid \land c.id = \$cid
\end{align*}
} We refer to the caller as the \emph{parent} function and the callee as the
\emph{child} function. Let $p$ and $c$ be the traces of events inside the parent
and child function bodies respectively. The join predicate $p.enter < c.enter
\land c.enter < p.exit$ selects calls to the child function from inside the
parent function. The predicate $p.id = \$pid \land c.id = \$cid$ selects the
pair of functions that we are interested in, where $\$pid$ and $\$cid$ are
parameters.

\subsection{The Layout Algebra}
\Projname{} programs are written in a language called the \langname{}. The
\langname{} is similar to the relational algebra, but as we will see shortly, it
can represent the layout of data as well as the operation of queries. By design,
it is more procedural than SQL, which is more akin to the relational
calculus~\cite{Codd1971}. For example, SQL leaves choices like join ordering to
the query planner, whereas in the \langname{} join ordering is explicit.

In designing the \langname{}, we follow a well-worn path in deductive synthesis
of creating a uniform representation that can capture all the refinement steps
from a high-level program to a low-level one. Accordingly, the \langname{} can
express programs which contain a mixture of high-level relational constructs and
low-level layout constructs. At some point, a \langname{} program contains
enough implementation information that the compiler can process it. We say that
these programs are \emph{serializable} (\Cref{sec:serializability}).

Here is the nested call query from \Cref{sec:background} translated into the
\langname{}: {\small
\begin{align*}
  &\select{}(\{enter_p, enter_c\},\ \join{}(enter_p < enter_c \land enter_c < exit_p,\\
  &\quad\filter{}(\$pid = id_p,\\
  &\quad\quad\select{}(\{id \mapsto id_p, enter \mapsto enter_p, exit \mapsto exit_p\}, log)), \\
  &\quad\filter{}(\$cid = id_c, \select{}(\{id \mapsto id_c, enter \mapsto enter_c\}, log))))
\end{align*}
} 
\vspace*{-3ex}

This program can be read as follows. $\textsf{filter}$ takes a predicate as its
first argument and a query as its second. It filters the query by the predicate.
$\textsf{join}$ takes a predicate as its first argument and queries as its
second and third. The two queries are joined together using the predicate.
$\textsf{select}$ takes a list of expressions with optional names and a query,
and selects the value of each expression for each tuple in the query, possibly
renaming it.

The scoping rules for the \langname{} may look somewhat unusual, but they are
intended to mimic the scoping conventions of SQL.\@ In this query, the names
$enter$, $exit$ and $id$ are field names in the $log$ relation. The \select{}
operators introduce new names for these fields, using the $\mapsto$ operator, so
that $log$ can be joined with itself. We formalize the semantics of the
\langname{}, including the scoping rules, in \Cref{sec:semantics}.

Note that at this point no layout is specified for $log$, so this is not a
serializable program. Even though it is not serializable (and so cannot be
compiled), this program still has well defined semantics. In
\Crefrange{sec:size-opt}{sec:both-opt}, we describe how layouts can be
incrementally introduced by transforming the program until it is serializable.

\subsection{Optimization Trade-offs}

The nested call query is interesting because the data in question is fairly
large---hundreds of thousands of rows---and keeping it fully in memory, or even
better in cache, is a significant performance win. Therefore, minimizing the
size of the data in memory should improve performance.

However, there is a fundamental trade-off between a more compact data
representation and allowing for efficient access. Sometimes the two goals are
aligned, but often they are not. For example, creating a hash index allows
efficient access using a key, but introduces overhead in the form of a mapping
between hash keys and values.

In the rest of this section we examine three layouts at different points in this
trade-off space: a compact nested layout with no index structures
(\Cref{fig:nested-layout}), a layout based on a single hash index
(\Cref{fig:hash-layout}), and a layout based on a hash index and an ordered
index (\Cref{fig:ordered-layout}). A priori, none of these layouts is clearly
superior. The hash based layout is the largest, but has the best lookup
properties. The nested layout precomputes the join and uses nesting to reduce
the result size, but is more expensive for lookups. The last layout must compute
the join at runtime but it has indexes that will make that computation fast. The
power of \projname{} is that it allows users to effectively explore different
layout trade-offs by freeing them from the necessity to ensure the correctness
of each candidate.

\subsection{Nested Layout}\label{sec:size-opt}
Our first approach to optimizing the nested call query is to materialize the
join, since joins are usually expensive, and to use nesting to reduce the size
of the resulting layout.

The first step is to apply transformation rules (\Cref{sec:relational-xform})
to hoist and merge the filters. Now the join is in a term with no query
parameters, so it can be evaluated at compile time: {\small
\begin{align*}
  &\select{}(\{enter_p, enter_c\},\ \filter{}(\$pid = id_p \land \$cid = id_c,\\
  &\quad\join{}(enter_p < enter_c \land enter_c < exit_p, \\
  &\quad\quad\select{}(\{id \mapsto id_p, enter \mapsto enter_p, exit \mapsto exit_p\}, log), \\
  &\quad\quad\select{}(\{id \mapsto id_c, enter \mapsto enter_c\}, log))))
\end{align*}
\vspace{-3ex}
} 

After applying two more rules---projection to eliminate unnecessary fields
(\Cref{sec:xform-proj}) and join elimination (\Cref{sec:xform-join})---the
result is the \emph{layout expression} represented in code in
\Cref{fig:nested-code} and graphically in \Cref{fig:nested-layout}. In this
program we see our first layout operators: $\alist{}(\cdot, \cdot)$ and
$\ctuple{[\dots]}$.\footnote{As a point of notation, we separate layout
  operators from non-layout operators visually by \textbf{bolding} them. This is
  just to make the programs easier to read.} The \langname{} extends the
relational algebra with these operators, allowing us to write layout
expressions, which describe how their arguments will be laid out in memory.

\begin{figure}
{\small
  \begin{align*}
    &\select{}(\{enter_p, enter_c\},\ \filter{}(id_c = \$cid \land id_p = \$pid,\\
    &\quad\alist{}(\as{\select{}(\{id \mapsto id_p, enter \mapsto enter_p, exit \mapsto exit_p\}, log)}{lp},\\
    &\quad\quad\ctuple{[\scalar{lp.id_p}, \scalar{lp.enter_p},\\
    &\quad\quad\quad\alist{}(\as{\filter{}(lp.enter_p < enter_c \land enter_c < lp.exit_p,\\
    &\quad\quad\quad\quad\select{}(\{id \mapsto id_c, enter \mapsto enter_c\}, log))}{lc},\\
    &\quad\quad\quad\quad\ctuple{[\scalar{lc.id_c}, \scalar{lc.enter_c}]})]})))
  \end{align*}
} 
\vspace{-4ex}
\caption{Nested list layout.}\label{fig:nested-code}
\vspace*{-3ex}
\end{figure}

\Cref{fig:nested-code} can be read as follows. The operator
$\alist{}(\as{log}{l}, q)$ creates a list with one element for every tuple in
$log$. Each element in the list is laid out according to $q$. For example, the
outermost list in \Cref{fig:nested-code} creates a list from the $id$, $enter$
and $exit$ fields of the $log$ relation. Each row in the list is laid out as a
$\atuple{\textsf{cross}}$\footnote{In this expression, \textsf{cross} specifies how the tuple
  will eventually be read. Layout operators evaluate to sequences, so a tuple
  needs to specify how these sequences should be combined. In this case, we take
  a cross product.}. The first two elements in the tuple are the scalar
representations of the $id_p$ and $enter_p$ fields, and the third element is a
nested list. Note that the content of that inner list is filtered based on the
value of $lp.enter_p$ and $lp.exit_p$, and each element laid out as just a pair
of two scalars $id_c$ and $enter_c$.

The query is now serializable because it satisfies a set of rules described in
\Cref{sec:serializability}. At a high-level, the rules require that we never use
a relation without specifying its layout, a requirement that in this case is
satisfied because all references to the original log relation appear in the
first arguments of \alist{} operators.

\Cref{fig:nested-layout} shows the structure of the resulting layout. This
layout is quite compact. It is smaller than the fully materialized join because
of the nesting; the caller \texttt{id} and \texttt{enter} fields are only stored
once for each matching callee record. When we benchmark this query, we find
found that it performs reasonably well (11.5ms) and is fairly small (50Mb).

We can make this layout more compact by applying further transformations. For
example, we know that $enter_p < enter_c < exit_p$. If we store $enter_c -
enter_p$ instead of $enter_c$ we can save some space by using fewer bits. The
ability to take advantage of this kind of knowledge about the structure of the
data is an important feature of our approach.

\begin{figure}
  \vspace*{-2ex}
  \centering
  \begin{minipage}{0.43\columnwidth}
  \begin{subfigure}[b]{\textwidth}
    \includegraphics[width=\textwidth]{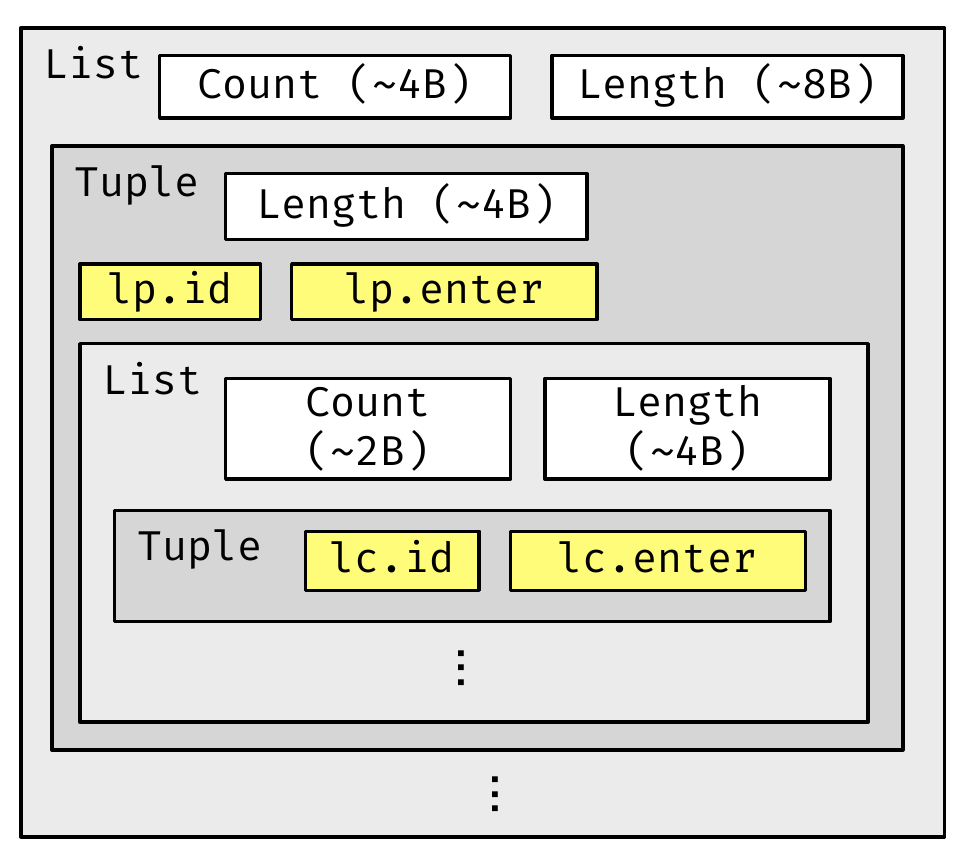}
    \caption{Nested.}\label{fig:nested-layout}
  \end{subfigure}
  \begin{subfigure}[b]{\textwidth}
    \includegraphics[width=\textwidth]{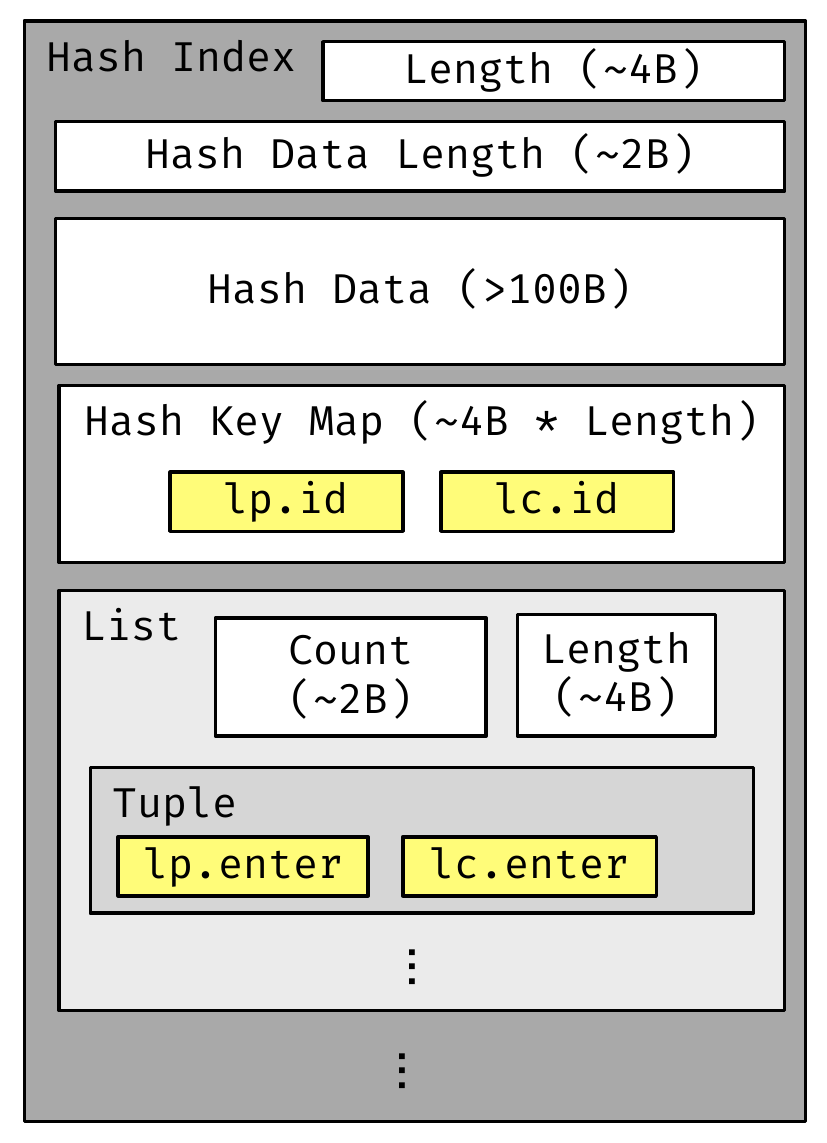}
    \caption{Hash-index.}\label{fig:hash-layout}
  \end{subfigure}
\end{minipage}%
\begin{minipage}{0.45\columnwidth}
\hspace{1em}%
  \begin{subfigure}[b]{\textwidth}
    \includegraphics[width=\textwidth]{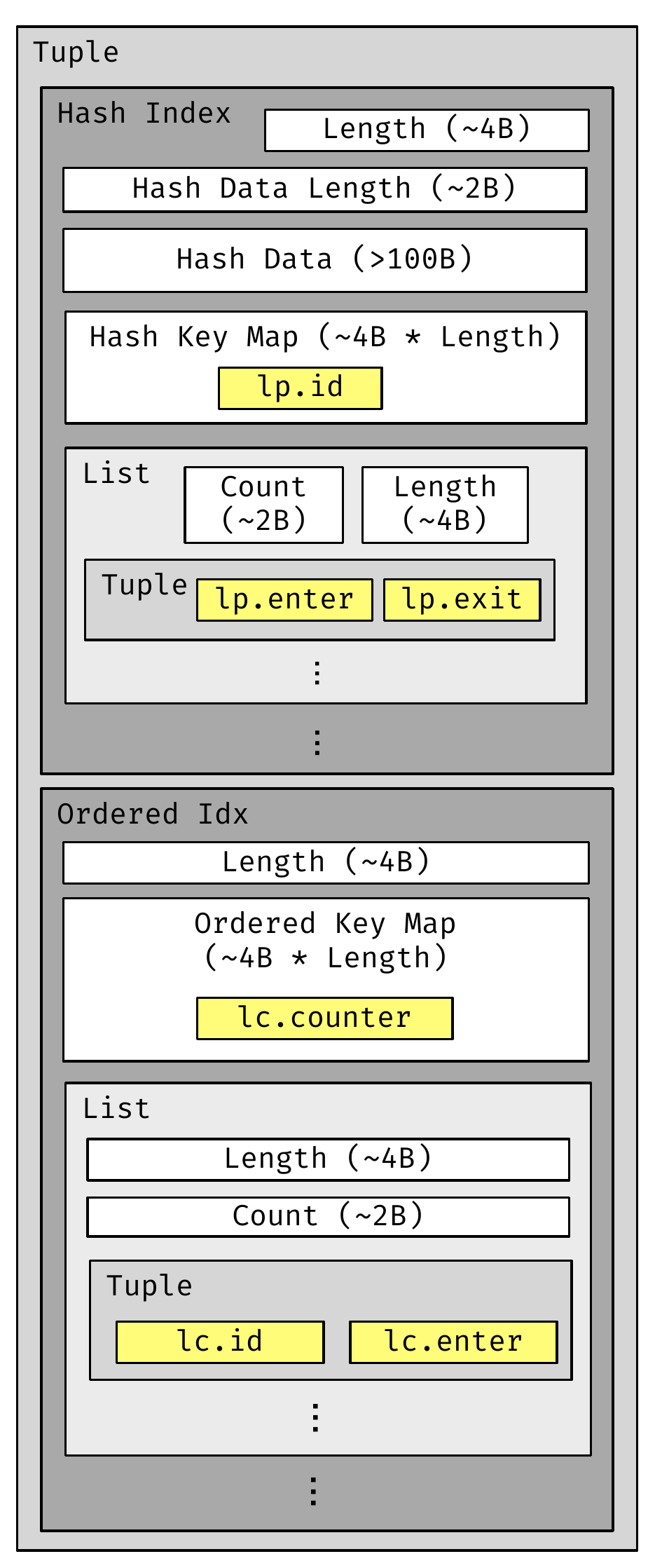}
    \caption{Ordered-index.}\label{fig:ordered-layout}
  \end{subfigure}
\end{minipage}
\vspace*{-2ex}
  \caption{Sample layouts (relational data is highlighted).}
  \vspace*{-3ex}
\end{figure}

\subsection{Hash-index Layout}\label{sec:lookup-opt}

Now we optimize for lookup performance by fully materializing the join and
creating a hash index. This layout will be larger than the nested layout but
look ups into the hash index will be quick, which will make evaluating the
equality predicates on $id$ fast. \Cref{fig:hash-layout} shows the structure of
the resulting layout. When we evaluate the query, we find that it is much faster
(0.4ms) but is larger than the nested query (60Mb).

\subsection{Hash- and Ordered-index Layout}\label{sec:both-opt}

Finally, we investigate a layout which avoids the full join materialization, but
still has enough indexing to be fast. We can see that the join condition is a
range predicate, so we would like to use an index that supports efficient range
queries to make that predicate efficient (\Cref{sec:xform-part}). Then we can
push the filters and introduce a hash table to select $id_p$. The resulting
layout is shown in \Cref{fig:ordered-layout}. This layout will be larger than
the original relation, but smaller than the other two layouts (9.8Mb), and it
allows for much faster computation of the join and one of the filters (0.6ms).

This program introduces three new operators: \oidx{}, \hidx{} and \depjoin{}.
\oidx{} creates indexes that support efficient range queries. Its first argument
defines the set of keys and its second argument defines the values and their
layout as a function of the keys. The remaining arguments are the upper and
lower bounds to use when reading the index ($p.enter_p$ and $p.exit$ in this
case). \hidx{} is similar, but it creates efficient point indexes (in this
query, $\$pid$ is the key). {\small
  \begin{align*}
    &\select{}(\{enter_p, enter_c\},\ \depjoin{}(\as{\hidx{}(\as{\select{}(\{id\},\ log)}{h},\\
    &\quad\alist{}(\as{\filter{}(h.id = id \land enter > exit,\ log)}{lh},\\
    &\quad\quad\ctuple{[\scalar{lh.enter \mapsto enter_p},\ \scalar{lh.exit}]}),\\
    &\quad\$pid)}{p},\\
    &\quad\filter{}(id = \$cid,\ \oidx{}(\as{\select{}(\{enter\}, log)}{o},\\
    &\quad\quad\alist{}(\as{\filter{}(enter = o.enter,\ log)}{lo},\\
    &\quad\quad\quad\ctuple{[\scalar{lo.id},\ \scalar{lo.enter \mapsto enter_c}]}),\\
    &\quad\quad p.enter_p,\ p.exit))))
  \end{align*}
}\vspace{-3ex}

The more interesting operator is \depjoin{}, which implements a dependent join.
A dependent join is one where the right-hand-side of the join can refer to
fields from the left-hand-side. In the \depjoin{} operator, the left-hand-side
is given a name (here it is $p$) that the right-hand-side can use to refer to
its fields. One way to think about a dependent join is as a relational for loop:
it evaluates the right-hand-side for each tuple in the left-hand-side,
concatenating the results. Unlike the layout operators \alist{}, \hidx{} and
\oidx{}, \depjoin{} executes entirely at runtime. It does not introduce any
layout structure.


\section{Language}

In this section we describe the \langname{} in detail. The \langname{} starts
with the relational algebra and extends it with layout operators. These layout
operators have relational semantics, but they also have layout semantics which
describes how to serialize them to data structures. The combination of
relational and layout operators allows the \langname{} to express both a query
and the data store that supports the execution of the query.

Programs in the \langname{} have three semantic interpretations:
\begin{enumerate}
\item The \emph{relational semantics} describes the behavior of a \langname{}
  program at a high level. We define this semantics using a theory of ordered
  finite relations~\cite{Cheung2013}. According to this semantics, A \langname{}
  program can be evaluated to a relation in a context containing relations and
  query parameters.
\item The \emph{layout semantics} describes how the compiler creates a data file
  from a serializable \langname{} program. The layout semantics operates in a
  context which contains relations, but \emph{not} query parameters.
\item The \emph{runtime semantics} describes how the compiled query executes,
  reading the layout file and using the query parameters to produce the query
  output. The runtime semantics operates in a context which contains query
  parameters but \emph{not} relations.
\end{enumerate}
These three semantics are connected: the layout semantics and the runtime
semantics combine to implement the relational semantics. The relational
semantics serves as a specification. An interpreter written according to the
relational semantics should execute \langname{} programs in the same way as our
compiler. In this section, we discuss the relational semantics in detail.
We leave the detailed discussion of the layout and runtime semantics to
\Cref{sec:layout-semantics} and \Cref{sec:runtime-semantics}.

\subsection{Syntax}\label{sec:syntax}

\begin{figure}
  {\small
    \begin{alignat*}{3}
      x &::=&&\ \text{identifiers} \quad o ::= \textsf{asc} ~|~
      \textsf{desc} \quad \tau ::= \textsf{cross} ~|~ \textsf{concat} \\
      v &::=&&\ \text{integers} ~|~ \text{strings} ~|~ \text{Booleans} ~|~ \text{floats} ~|~ \text{dates} ~|~ \textsf{null} \\
      e &::=&&\ v ~|~ x ~|~ e + e' ~|~ e - e' ~|~ e \times e' ~|~ e / e' ~|~ e\
      \%\ e' ~|~ e < e' ~|~ e \leq e' ~|~ e > e' \\
      &\ccol{|}&&\ e \geq e' ~|~ e = e' 
      ~|~ \ite{e}{e_t}{e_f} ~|~ \textsf{exists}(q) ~|~ \textsf{first}(q) ~|~
      \textsf{count}() \\
      &\ccol{|}&&\ \textsf{sum}(e) ~|~ \textsf{min}(e)
      ~|~ \textsf{max}(e) ~|~ \textsf{avg}(e) \\
      t &::=&&\ \{x_1 \mapsto e_1, \dots, x_k \mapsto e_k\} \\
      q &::=&&\ \textsf{relation}(x) ~|~ \textsf{dedup}(q) ~|~
      \textsf{select}(t, q) ~|~ \textsf{filter}(e, q) ~|~ \textsf{join}(e, q,
      q') \\
      &\ccol{|}&&\ \textsf{group-by}(t, [x_1, \dots, x_m], q) ~|~
      \textsf{order-by}([e_1\ o_1, \dots, e_m\ o_m], q) \\
      &\ccol{|}&&\  \textsf{depjoin}(\as{q}{x}, q') ~|~ \scalar{x \mapsto e}
      ~|~\atuple{\tau}(t) ~|~ \alist{}(\as{q_r}{x}, q) \\
      &\ccol{|}&&\ \hidx{}(\as{q_k}{x}, q_v, t_k) ~|~ \oidx{}(\as{q_k}{x}, q_v, t_{lo}, t_{hi})~|~\emptyset \\
    \end{alignat*}
  } 
  \vspace*{-7ex}
  \caption{Syntax of the \langname.}\label{fig:syntax}
  \vspace*{-3ex}
\end{figure}


\Cref{fig:syntax} shows the syntax of the \langname{}. Note that the \langname{}
can be divided into relational operators (\select{}, \filter{}, \join{}, etc.)
and layout operators (\alist{}, \hidx{}, etc.). The \langname{} is a strict
superset of the relational algebra. In fact, the layout operators have
relational semantics in addition to byte-level data layout semantics
(see~\Cref{sec:layouts}).

\subsection{Semantics}\label{sec:semantics}

\begin{figure}
  \vspace{-1ex}
  \begin{small}
  \begin{gather*}
    Id = (Scope?, Name) \quad
    Context = Tuple = \{Id \mapsto Value\} \\
    Relation = [Tuple] \\
    \sigma : Context \enspace
    \delta : Id \mapsto Relation \enspace
    s : Id \enspace
    t : Tuple \enspace
    v : Value \enspace
    r : Relation
  \end{gather*}
  \begin{equation}
    \inference{
      \forall n.\ n \not \in \sigma \lor n \not \in \sigma' 
    }{
      \sigma \cup \sigma' = \{n_i \mapsto v_i ~|~ \sigma[n_i] = v_i \lor \sigma'[n_i] =
      v_i\}
    }\label{eq:ctx}
  \end{equation}
  \vspace{1ex}
  \begin{equation}
    \inference{
      t = \{n_1 \mapsto e_1, \dots, n_m \mapsto e_m\} &
      \forall i.\ \eto{e_i}{v_i} 
    }{
      \eto{t}{\{n_1 \mapsto v_1, \dots, n_m \mapsto v_m\}}
    }\label{eq:record}
  \end{equation}
  
  \begin{equation}
    \inference{
      \delta[n] = r
    }{
      \eto{\textsf{relation}(n)}{r}
    }\label{eq:relation}
  \end{equation}
  \vspace{1ex}
  \begin{equation}
    \inference{
      \eto{q}{r_q} \\
      r = [t ~|~ t <- r_q \quad \eto[\sigma \cup t, \delta]{e}{\textsf{true}}]
    }{
      \eto{\textsf{filter}(e, q)}{r}
    }\label{eq:filter}
  \end{equation}
  \vspace{1ex}
  \begin{equation}
    \inference{
      \eto{q}{r} &
      r'' = \left[t' ~\Big|~
        {\begin{array}{c}
           t <- r,\ t' <- r' \quad \eto[\sigma \cup t_s, \delta]{q'}{r'} \\
           t_s = \{s.f \mapsto v ~|~ (f \mapsto v) \in t\}
         \end{array}}
     \right]
   }{
     \eto{\textsf{depjoin}(\as{q}{s}, q')}{r''}
   }\label{eq:depjoin}
 \end{equation}

 \begin{minipage}{0.65\columnwidth}
   \begin{equation}
     \inference{\eto{e}{v}}{
       \eto{\scalar{n \mapsto e}}{[\{n \mapsto v\}]}
     }\label{eq:scalar}
   \end{equation}
 \end{minipage}%
 \begin{minipage}{0.35\columnwidth}
   \begin{equation}
     \inference{}{\sigma, \delta |- \emptyset \evalto [\ ]}\label{eq:empty}
   \end{equation}
 \end{minipage}
 
 \begin{minipage}{0.55\columnwidth}
   \begin{equation}
     \inference{
       \eto{\textsf{depjoin}(\as{q_r}{s}, q)}{r}
     }{
       \eto{\alist{}(\as{q_r}{s}, q)}{r}
     }\label{eq:list}
   \end{equation}
 \end{minipage}%
 \begin{minipage}{0.45\columnwidth}
   \begin{equation}
     \inference{}{\eto{\atuple{\tau}([~])}{[~]}}\label{eq:etuple}
   \end{equation}
 \end{minipage}

 \begin{equation}
   \inference{
     \eto{q_1}{r_q} &
     \eto{\atuple{\tau}([q_2, \dots, q_n])}{r_{qs}} 
   }{
     \eto{\atuple{\textsf{cross}}([q_1, \dots, q_n])}{[t \cup ts ~|~ t <- r_{q}, ts <- r_{qs}]}
   }\label{eq:ctuple}
 \end{equation}
 
 \begin{equation}
   \inference{
     \eto{\textsf{depjoin}(\as{q_k}{s}, \filter{}(l_{lo}\leq s.x \leq l_{hi},
       q_v))}{r} \\
     \textsc{schema}(q_k) = [x]
   }{
     \eto{\oidx{}(\as{q_k}{s}, q_v, l_{lo}, l_{hi})}{r}
   }\label{eq:oidx}
 \end{equation}
\end{small}
\vspace{-2ex}
\caption{Selected relational semantics of the \langname.}\label{fig:runtime-semantics}
\vspace{-2ex}
\end{figure}


The semantics (\Cref{fig:runtime-semantics}) operates on three kinds of values:
scalars, tuples and relations. Scalars are values like integers, Booleans, and
strings. Tuples are finite mappings from field names to scalar values. Relations
are represented as finite, ordered sequences of tuples. $[~]$ stands for the
empty relation, $:$ is the relation constructor, and $\concat$ denotes the
concatenation of relations.

We use sequences instead of sets for two reasons. First, sequences are more like
bag semantics than the set semantics of the original relational algebra. This
choice brings the \langname{} more in line with the semantics of SQL, which is
convenient for our implementation. Second, sequences allow us to represent query
outputs which have an ordering.

In the semantic rules, $\sigma$ is an evaluation context; it maps names to
scalar values. $\delta$ is a relational context; it maps names to relations. We
separate the two contexts because the relational context $\delta$ is global and
immutable; it consists of a universe of relations that exist when the query is
executed (or compiled) which are contained in some other database system. The
evaluation context $\sigma$ initially contains the query parameters, but some
operators introduce new bindings in $\sigma$. $\cup$ denotes the binding of a
tuple into an evaluation context. Read $\sigma \cup t$ as a new evaluation
context that contains the fields in $t$ in addition to the names already in
$\sigma$.

In the rules, $|-$ separates contexts and expressions and $\evalto$ separates
expressions and results. Read $\sigma, \delta |- l \evalto s$ as ``the layout
$l$ evaluates to the relation $s$ in the context $\sigma, \delta$.''

We borrow the syntax of list comprehensions to describe the semantics of the
\langname{} operators. For example, consider the list comprehension in the
\textsf{filter} rule:
$[t ~|~ t <- r_q, (\sigma \cup t, \delta |- e \evalto \textsf{true})]$, which
corresponds to the expression $\textsf{filter}(e,q)$. This list comprehension
filters $r_q$ by the predicate $e$ where $r_q$ is the relation produced by $q$.
$e$ is evaluated in a context $\sigma \cup t$ for each tuple $t$ in $r_q$.

Comprehensions that contain multiple $<-$, as in the \textsf{join} rule, should
be read as a cross product.

\subsubsection{Relational Operators}\label{sec:relational-ops}

First, we describe the semantics of the relational operators: \textsf{relation},
\textsf{filter}, \textsf{join}, \textsf{select}, \textsf{group-by},
\textsf{orderby}, \textsf{dedup}, and \textsf{depjoin}. These operators are
modeled after their equivalent SQL constructs. For brevity and because they are
straightforward, we omit the rules for selection with aggregates,
\textsf{group-by}, and \textsf{order-by} from \Cref{fig:runtime-semantics}.

\textsf{relation} returns the contents of a relation in the relational context
$\delta$. \textsf{filter} filters a relation by a predicate $e$. \textsf{join}
takes the cross product of two relations and filters it using a predicate $e$.

\textsf{select} is used for projection, aggregation, and renaming fields. It
takes a tuple expression $t$ and a relation $r$. If $t$ contains no aggregation
operators, then a new tuple will be constructed according to $t$ for each tuple
in $r$. If $t$ contains an aggregation operator (\textsf{count}, \textsf{sum},
\textsf{min}, \textsf{max}, \textsf{avg}), then \textsf{select} will aggregate
the rows in $r$. If $t$ contains both aggregation and non-aggregation operators,
then the non-aggregation operators will be evaluated on the last tuple in $r$.

\textsf{group-by} takes a list of expressions, a list of fields, and a relation.
It groups the tuples in the relation by the values of the fields, then computes
the aggregates in the expression list. \textsf{order-by} takes a list of
expression-order pairs and a relation. It orders the tuples in the relation
using the list of expressions to compute a key. \textsf{dedup} removes duplicate
tuples.

Finally, \textsf{depjoin} denotes a dependent join, where the right-hand-side of
the join can depend on values from the left-hand-side. It is similar to a
for-each loop; $\textsf{depjoin}(\as{q}{n},\ q')$\footnote{In this expression,
  $n$ is a \emph{scope}, and it qualifies the names in $q$. Scopes are discussed
  in more detail in \Cref{sec:scopes}.} can be read ``evaluate $q'$ for each
tuple in $q$ and concatenate the results.'' We use \depjoin{} as a building
block to define the semantics of the layout operators.

\subsubsection{Layout Operators}\label{sec:layouts}

We extend the relational algebra with layout operators that specify the layout
of data in memory at a byte level. The nesting and ordering of the layout
operators correspond to the nesting and ordering of the data structures that
they represent. Nesting allows data that is accessed together to be stored
together, increasing spatial locality. Note that layout operators can capture
the results of executing common relational algebra operations such as joins or
selections, allowing query processing to be performed at compile time. In
addition, layout primitives can express common relational data storage patterns,
such as row stores and clustered indexes.

\projname{} supports the following data structures:

\noindent
\textbf{Scalars:} Scalars can be machine integers (up to 64 bits), strings,
Booleans, and decimal fixed-point.

\noindent
\textbf{Tuples:} Tuples are layouts that can contain layouts with
  different types. If a collection contains tuples, all the tuples must have the
  same number of elements and their elements must have compatible types. Tuples
  can be read either by taking the cross product or concatenating their
  sub-layouts.
  
\noindent
\textbf{Lists:} Lists are variable-length layouts. Their contents must be
of the same type.

\noindent
\textbf{Hash indexes:} Hash indexes are mappings between scalar keys and
  layouts, stored as hash tables. Like lists, their keys must have the same
  type.
  
\noindent
\textbf{Ordered indexes:} Ordered indexes are mappings between scalar keys
and layouts, stored as ordered mappings.

Each data structure has a corresponding layout operator. The layout operators
are the novel part of the \langname{} and their semantics are therefore
non-standard. The relational semantics of the layout operators are
in~\Cref{fig:runtime-semantics}. Although the layout operators can be used to
construct complex, nested layouts, they evaluate to flat relations of tuples of
scalars, just like the relational operators. The rules
in~\Cref{fig:runtime-semantics} only describe the relational behavior of the
layout operators; they do not address the question of how data is laid out or
how it is accessed. We discuss these aspects of the layout operators
in~\Cref{sec:layout-semantics}.

The simplest layout operators are \nscalar{} and \ntuple{}. The \nscalar{}
represents a single scalar value. Evaluating a \nscalar{} operator produces a
relation containing a single tuple. The \ntuple{} operator represents a
fixed-size, heterogeneous list of layouts. When evaluated, each layout in the
tuple produces a relation, which are combined either with a cross product or by
concatenation.

Note that evaluating a \ntuple{} operator produces a \emph{relation} not a tuple.
Although these semantics are slightly surprising, there are two reasons why we
chose this behavior. First, it is consistent with the other layout operators,
all of which evaluate to relations. Second, \ntuple{}s can contain other
layouts (\alist{}s for example) which themselves evaluate to relations.

The remaining layout operators---\alist{}, \hidx{} and \oidx{}---have a similar
structure. We discuss the \alist{} operator in detail. \Cref{eq:list} specifies
the behavior of \alist{}. Note that \alist{} is essentially an alias for
\depjoin{}. Like \depjoin{}, \alist{} takes two arguments: $q_r$ and $q$. These
two arguments should be interpreted as follows: $q_r$ describes the data in the
list. Each element of the list has a corresponding tuple in $q_r$, so the length
of the list is the same as the length of $q_r$. One can think of each tuple in
$q_r$ as a kind of key that determines the contents of each list element. On the
other hand, $q$ describes how each list element is laid out. $q$ will be
evaluated separately for each tuple in $q_r$. It determines for each ``key'' in
$q_r$, what the physical layout of each list element will be, as well as how
that element must be read.

Returning to the query in~\Cref{sec:size-opt}, the inner \alist{}
operator
\begin{small}
\begin{align*}
  &\alist{}(\as{\filter{}(lp.enter_p < enter_c \land enter_c < lp.exit_p,\\
  &\quad\select{}(\{id \mapsto id_c,\ enter \mapsto enter_c\},\ log))}{lc},\\
  &\quad\quad\ctuple{[\scalar{lc.id_c},\ \scalar{lc.enter_c}]})
\end{align*}
\end{small}
selects the tuples in $log$ where $enter_c$ is between $enter_p$ and $exit_p$,
and creates a list of these tuples. The first argument describes the contents of
the list and the second describes their layout. This program will generate a
layout that has a list of tuples, structured as $[(id_{c1}, enter_{c1}), \dots,
(id_{cn}, enter_{cn})]$.

\hidx{} and \oidx{} are similar to \alist{}. They have a query $q_k$
that describes which keys are in the index and a query $q_v$ that describes the
contents and layout of the values in the index. For example, in:
\begin{small}
\begin{align*}
  &\hidx{}(\as{\select{}(\{id\}, log)}{h},\ \alist{}(\as{\filter{}(\{id = h.id\}, log)}{l},\\
  &\quad\ctuple{[\scalar{l.enter}, \scalar{l.exit}]}),\ \$pid),
\end{align*}
\end{small}
the keys to the hash-index are the $id$ fields from the $log$ relation. For each
of these fields, the index contains a list of corresponding $(enter, exit)$
pairs, stored in a tuple. When the hash-index is accessed, $\$pid$ is used as
the key. This program generates a layout of the form: $\{id \mapsto [(enter,
exit), \dots], \dots\}$, which is a hash-index with scalars for keys and lists
of tuples for values.

\subsubsection{Scopes \& Name Binding}\label{sec:scopes}

The scoping rules of the \langname{} are somewhat more complex than the
relational algebra. There are two ways to bind a name in the \langname{}: by
creating a relation or by using an operator which creates a scope.

All of the operators in the \langname{} return a relation. Some operators simply
pass through the names in their parameter relations. Others, such as
\textsf{select} and \nscalar{} can be used for renaming or for creating new
fields.

Some operators, such as \depjoin{}, create a \emph{scope}. A scope is a tag which
uniquely identifies the binding site of a name. For example, in
$\depjoin{}(\as{q}{s}, q')$, a field $f$ from $q$ is bound in $q'$ as $s.f$.
Scoped names with distinct scopes are distinct and scoped names are distinct
from unscoped names. We add scopes to the \langname{} as a syntactically
lightweight mechanism for renaming an entire relation. Renaming entire relations
is necessary because shadowing is prohibited in the \langname{}. Prohibiting
shadowing removes a major source of complexity when writing transformations.
While we could use \select{} for renaming, we opted to add scopes so that
renaming at binding sites would be part of the language rather than a pervasive
and verbose pattern.

There are still situations when renaming entire relations using \select{} is
necessary. For example, in a self-join one side of the join must be renamed.

\subsection{Staging \& Serializability}\label{sec:serializability}

Another way to view the three semantic interpretations is from the point of view
of multi-stage programming. A serializable \langname{} program can be evaluated
in two stages: the layout is constructed in a compile-time stage, then the
compiled query reads the layout and processes it in a run-time stage. However,
while traditionally program staging is used to implement code specialization, in
the \langname{} staging is used to implement data specialization. This
difference in focus leads to different implementation challenges. In particular,
the ``unstaged'' version of a \langname{} program is often large (tens to
hundreds of megabytes). The \langname{} compiler must be carefully designed to
handle this scale.

In \Cref{sec:layouts}, we explained how the layout operators execute in two
stages: one stage at compile time and one stage at query runtime. Only a subset
of \langname{} programs can be separated in this way. We say that programs which
can be properly staged are \emph{serializable}.

A program is serializable if and only if the names referred to in compile (resp.
run) time contexts are bound in compile (resp.\ run) time contexts. An expression
is in a compile-time context if it appears in the first argument to \alist{},
\hidx{}, \oidx{}, or \nscalar{}. Otherwise, it is in a run-time context. We
consider the relations in $\delta$ to be bound in a compile-time context and
query parameters to be bound in a run-time context. The compiler uses a simple
type system that tracks the stage of each name in the program to check for
serializability.

Transforming a program into a serializable form is a key goal of our automatic
optimizer (\Cref{sec:optimizer}). Many of the rules that the optimizer applies
can be seen as moving parts of the query between stages. 


\section{Transformations}\label{sec:rules}

In this section, we define semantics preserving transformation rules that
optimize query and layout performance. These rules change the behavior of the
program with respect to the layout and runtime semantics while preserving it
with respect to the relational semantics. These rules subsume standard query
optimizations because in addition to changing the structure of the query, they
can also change the structure of the data that the query processes.


\subsection{Notation}
Transformations are written as inference rules. When writing inference rules,
$e$ will refer to scalar expressions and $q$ will refer to \langname{}
expressions. $E$ and $Q$ will refer to lists of expressions and layouts. In
general, the names we use correspond to those used in the syntax description
(\Cref{fig:syntax}). If we need to refer to a piece of concrete syntax, it will
be formatted as e.g., $\textsf{concat}$ or $\textsf{x}$.

To avoid writing many trivial inductive rules, we define contexts. The
definition is straightforward, so we leave it to the appendix
(\Cref{fig:contexts}). If $C$ is a context and $q$ is a \langname{} expression,
then $C[q]$ is the expression obtained by substituting $q$ into the hole in $C$.
In addition to contexts, we define two operators: $\cxform{}$ and $->$. $q
\cxform{} q'$ means that the \langname{} expression $q$ can be transformed into
$q'$ and $q -> q'$ means that $q$ can be transformed into $q'$ in any context.
The relationship between these two operators is:
\begin{small}\[q -> q' \equiv \forall C.\ C[q] \cxform{} C[q']\]\end{small}
\vspace{-3ex}

\subsection{Relational Optimization}\label{sec:relational-xform}
There is a broad class of query transformations that have been developed in the
query optimization literature~\cite{Jarke1984,Chaudhuri1998}. These
transformations can generally be applied directly in \projname{}, at least to
the relational operators. For example, commuting and reassociating joins, filter
pushing and hoisting, and splitting and merging filter and join predicates are
implemented in \projname{}. Although producing optimal relational algebra
implementations of a query is explicitly a non-goal of \projname{}, these kinds
of transformations are important for exposing layout optimizations.

\subsection{Projection}\label{sec:xform-proj}
Projection, or the removal of unnecessary fields from a query, is an important
transformation because many queries only use a small number of fields; the most
impactful layout specialization that can be performed for these queries is to
remove unneeded fields.

First, we need to decide what fields are necessary. For a query $q$ in some
context $C$, the necessary fields in $q$ are visible in the output of $C[q]$ or
are referred to in $C$. Let $\textsc{schema}(\cdot)$ be a function from a layout
$q$ to the set of field names in the output of $q$. Let $\textsc{names}(\cdot)$
be a function which returns the set of names in a context or layout expression.
Let $\textsc{needed}(\cdot,\cdot)$ be a function from contexts $C$ and layouts
$q$ to the set of necessary fields in the output of $q$:
\[
  \textsc{needed}(C, q) = \textsc{schema}(q) \cap (\textsc{schema}(C[q])
  \cup\textsc{names}(C))
\]

$\textsc{needed}(\cdot, \cdot)$ can be used to define transformations which
remove unnecessary parts of a layout. For example, this rule removes unnecessary
fields from tuples:
\begin{small}
\[
  \inference{
    Q' = [q ~|~ q\in Q, \textsc{needed}(C, q) \neq \emptyset] 
  }{
    C[\atuple{}(Q)] \cxform{} C[\atuple{}(Q')]
  }
\]
\end{small}
There is a similar rule for \textsf{select} and \textsf{groupby} operators.

The projection rules differ from the others in this section because they refer
to the context $C$. The other rules can be applied in any context. The context
is important for the projection rules because without it, all the fields in a
layout would be visible and therefore ``necessary''. Referring to the context
allows us to determine which fields are visible to the user.

\subsection{Precomputation}\label{sec:xform-precomp}
A simple transformation that can improve query performance is to compute and
store the values of parameter-free terms. This transformation is similar to
partial evaluation. The following rule\footnote{Some of the rules make a
  distinction for \emph{parameter-free} expressions, which do not contain query
  parameters. In these rules, parameter-free expressions are denoted as
  $\static{e}$. } precomputes a static \langname{} expression:
\begin{small}
\[
  \inference{
    \textsc{schema}(\static{q}) = [f_1, \dots, f_k] & x\ \text{is fresh}
  }{
    \static{q} -> \alist{}(\as{\static{q}}{x}, \ctuple{[\scalar{x.f_1}, \dots, \scalar{x.f_k}]})
  }
\]
\end{small}

Hoisting static expressions out of predicates can also be very profitable:
\begin{small}
\[
  \inference{
    x, y\ \text{are fresh} \quad
    \static{e'}\ \text{is a term in}\ e \quad
    \textsc{names}(e') \cap \textsc{schema}(q) = \emptyset
  }{
    \filter{}(e, q) -> \textsf{depjoin}(\as{\scalar{e' \mapsto y}}{x}, \filter{}(e[e' := x.y], q))
}.
\]
\end{small}
If the expression $e'$ can be precomputed and stored instead of being recomputed
for every invocation of the filter. Similar transformations can be applied to
any operator that contains an expression. This rule is useful when the filter
appears inside a layout operator. For example, in $\alist{}(\as{q}{x},
\textsf{filter}(e, q'))$, an expression $e'$ can be hoisted out of the filter if
it refers to the fields in $q$ but not if it refers to the fields in $q'$.

In a similar vein, \select{} operators can be partially precomputed. For
example:
\begin{small}
\[
  \inference{
    y'\ \text{is fresh} & q_v' = \select{}(\{\textsf{sum}(e) \mapsto y'\}, q_v)
  }{
    \begin{aligned}
    &\select{}(\{\textsf{sum}(e) \mapsto y\}, \oidx{}(\as{q_k}{x}, q_v, t_{lo}, t_{hi})) -> \\
    &\quad\select{}(\{\textsf{sum}(y') \mapsto y\}, \oidx{}(\as{q_k}{x}, q_v', t_{lo}, t_{hi}))
  \end{aligned}
}.
\]
\end{small}
After this transformation, the ordered index will contain partial sums which
will be aggregated by the outer select. This rule is particularly useful when
implementing grouping and filtering queries, because the filter can be replaced
by an index and the aggregate applied to the contents of the index. A similar
rule also applies to \select{} and \alist{}. A simple version of this rule
applies to \hidx{}; in that case, the outer \select{} is unnecessary.

This transformation is combined with group-by elimination
(\Cref{sec:xform-part}) in TPC-H query 1 to construct a layout that precomputes
most of the aggregation.

\subsection{Partitioning}\label{sec:xform-part}

Partitioning is a fundamental layout transformation that splits one layout into
many layouts based on the value of a field or expression. A partition of a
relation $r$ is defined by an expression $e$ over the fields in $r$. Tuples in
$r$ are in the same partition if and only if evaluating $e$ over their fields
gives the same value.

Let $\partition{\cdot}{\cdot}{\cdot}$ be a function which takes a layout $q$, a
partition expression $e$, and a name $x$, and returns a query for the partition
keys and a query for the partitions:
\begin{small}
\[
  \partition{q}{\static{e}}{x} = (\textsf{dedup}(\textsf{select}(\static{e},
    q)),\ \textsf{filter}(x.e = \static{e}, q)).
\]
\end{small}
In this definition, $q_k$ evaluates to the unique valuations of $e$ in $r$. These are
the partition keys. Note that the expression $q_v$ contains a free scope $x$. We
use $x.e$ to denote the expression $e$ with its names qualified by the scope
$x$. Once $x.e$ is bound to a particular partition key, $q_v$ evaluates to a
relation containing only tuples in that partition.

The partition function is used to define rules that create hash indexes and
ordered indexes from filters:
\begin{small}
\[
  \inference{
    x, n\ \text{is fresh} &
    \partition{q}{\static{e}}{x} = (q_k, q_v)
  }{
    \textsf{filter}(\static{e} = e', q) -> \hidx{}(\as{q_k}{x}, q_k, q_v, e')
  },
\]
\end{small}
\begin{small}
\[
  \inference{
    x\ \text{is fresh} &
    \partition{q}{\static{e}}{x} = (q_k, q_v)
  }{
    \textsf{filter}(e_l \leq \static{e} \land \static{e} \leq e_h, q) -> \oidx{}(\as{q_k}{x}, q_k, q_v, e_l, e_h)
  }.
\]
\end{small}

Partitioning also leads immediately to a rule that eliminates
$\groupby{}(\cdot)$:
\begin{small}
\[
  \inference{
    x\ \text{is fresh} &
    \partition{q}{\static{K}}{x} = (q_k, q_v)
  }{
    \groupby{}(\static{K}, E, q) -> \alist{}(\as{q_k}{x}, \select{}(E, q_v))
  }.
\]
\end{small}
There is a slight abuse of notation in this rule. $K$ is a list of expressions,
so the filter in $q_v$ must have an equality check for each expression in $K$.
This group-by elimination rule is used in many of the TPC-H queries which
contain \groupby{}s.

\subsection{Join Elimination}\label{sec:xform-join}

\projname{}'s layout operators admit several options for join materialization.
Since joins are often the most expensive operations in a relational query,
choosing a good join materialization strategy is critical. \projname{} does not
suggest a join strategy but it provides the necessary tools for an expert user.

Partitioning can be used to implement join materialization: a powerful
transformation that can significantly reduce the computation required to run a
query, at the cost of increasing the size of the data that the query runs on.
Our layout language allows for several join materialization strategies.

For example, joins can be materialized as a list of pairs:
\begin{small}
\[
  \inference{
    x\ \text{is fresh} \quad
    \partition{q}{\static{e}}{x} = (q_k, q_v) \quad
    \partition{q'}{\static{e'}}{x} = (\cdot, q_v')
  }{
    \textsf{join}(\static{e} = \static{e'}, q, q') ->
    \alist{}(\as{q_k}{x}, \ctuple{[q_v, q_v']})
  }.
\]
\end{small}
Each pair in this layout contains the tuples that should join together from the
left- and right-hand-sides of the join.

Joins can also be materialized as nested lists:
\begin{small}
\[
  \inference{
    x\ \text{is fresh} \enspace
    \textsc{schema}(q) = [f_1, \dots, f_n] \enspace
    F = \scalar{f_1}, \dots, \scalar{f_n}
  }{
      \textsf{join}(e = e', q, q') -> 
      \alist{}(\as{q}{x},\ \ctuple{[F,\ \filter{}(x.e = e', q')]})
  }.
\]
\end{small}
This layout works well for one-to-many joins, because it only stores each
row from the left hand side of the join once, regardless of the number of
matching rows on the right hand side.

Or, joins can be materialized as a list and a hash table:
\begin{small}
\[
  \inference{
    x, x'\ \text{are fresh} & \partition{q'}{\static{e'}}{x'} = (q_k, q_v) & k = x.e
  }{
    \textsf{join}(e = \static{e'}, q, q') -> \textsf{depjoin}(\as{q}{x},
    \hidx{}(\as{q_k}{x'}, q_k, q_v, k))
}.
\]
\end{small}
This is similar to how a traditional database would implement a hash join, but
in our case the hash table is precomputed. Using a hash table adds some overhead
from the indirection and the hash function but avoids materializing the cross
product if the join result is large.

If the join is many-to-many with an intermediate table, then either of the
above one-to-many strategies can be applied.


\subsection{Predicate Precomputation}\label{sec:xform-pred-precomp}
In some queries, it is known in advance that a parameter will come from a
restricted domain. If this parameter is used as part of a filter or join
predicate, precomputing the result of running the predicate for the known
parameter space can be profitable, particularly when the predicate is expensive
to compute. Let $p$ be a query parameter and $D_p$ be the domain of values that
$p$ can assume.
\begin{small}
\[
  \inference{
    \textsc{relations}(q) = \{r\} & \textsc{params}(e) = \{p\} &
    w_i = e[p := D_p[i]] \\
    e' = \bigvee_{i} (w_i \land p = D_p[i]) \lor e &
    r'=\select{}([w_1, \dots, w_{|D_p|}, \dots], r)
  }{
    \filter{}(e, q) -> \filter{}(e', q[r := r'])
  }
\]
\end{small}
This rule generates an expression $w_i$ for each instantiation of the predicate
with a value from $D_p$. The $w_i$s are selected along with the original
relation $r$. When we later create a layout for $r$, the $w_i$s will be stored
alongside it. When the filter is executed, if the parameter $p$ is in $D_p$, the
or will short-circuit and the original predicate will not run. However, this
transformation is semantics preserving even if $D_p$ is underapproximate. If the
query receives an unexpected parameter, then it executes the original predicate
$e$. Note that in the revised predicate $e'$, $p = D_P[i]$ can be computed once
for each $i$, rather than once per invocation of the filter predicate.

We use this transformation on TPC-H queries 2 and 9 to eliminate expensive
string comparisons.

\subsection{Correctness}

To show that the semantics that we have outlined in \Cref{sec:semantics} are
sufficient to prove the correctness of nontrivial transformations, we prove the
correctness of the equality filter elimination rule (\Cref{sec:xform-part}) in
\Cref{sec:partition-proof}. Although we do not prove the correctness of all of
the rules, this example demonstrates that such proofs are possible.

In particular, since our notation mixes relational and layout constructs, even
transformations that manipulate both the run- and compile-time behavior of the
query are often local transformations, and are therefore simple to prove
correct.


\section{Compilation}\label{sec:compiler}

The result of applying the transformation rules is a program in the layout
algebra. This program is still quite declarative, so there is a significant
abstraction gap to cross before the program can be executed efficiently.
Compilation of \langname{} programs proceeds in three passes:

\stitle{Type Inference.} The type inference pass computes a \emph{layout type},
which contains information about the ranges of values in the layout. For
example, integers are abstracted using intervals, as are the numerators of fixed
point numbers. Note that every element in collections like lists and indexes
must be of the same type but tuples can contain elements of different types.

\stitle{Serialization.} The serialization pass generates a binary representation
of the layout, using information from the layout type to specialize the layout
to the data. Each of the layout operators has a binary serialization format
which is intended to (1) take up minimal space and (2) minimize the use of
pointers to preserve data locality.

\stitle{Code generation.} Query code is generated according to the compilation
strategy described in~\cite{Tahboub2018}. This is referred to as push-based, or
data-centric query evaluation. We found that using this strategy instead of a
traditional iterator model is critical for query performance. A syntax-directed
lowering pass transforms each query and layout operator into an imperative
intermediate representation, using the layout type to generate the layout
reading code. This IR is then lowered to LLVM IR, optimized, and compiled into
an executable that provides a command line interface to the query.


\section{Optimization}\label{sec:optimizer}

\Projname{} includes an automatic, cost guided optimizer for the \langname{}.
Given a query written in the relational algebra fragment, the optimizer searches
for a sequence of transformations that (1) makes the query serializable
(\Cref{sec:serializability}) and (2) minimizes the cost of executing the query.
The optimizer consists of two components: a transformation scheduling language
and a cost model for the \langname{}.

\stitle{Scheduling.} The space of transformation sequences is far too large for
an exhaustive search, so we write a schedule that only considers a subset of the
full space of transformations. We use a small domain specific language to
construct this schedule. This language is inspired by~\cite{Visser2005} and
provides combinators for sequencing, fix-points, and context selection. The
schedule captures some of the domain knowledge that we have about how to
optimize query layouts.

The optimizer schedule has four phases: join nest elimination, hash-index
introduction, ordered-index introduction, and precomputation. These phases are
not the entirety of the optimizer but they give a rough picture of its behavior.

The join nest elimination phase looks for unparameterized join nests and
replaces them with layouts. As discussed in \Cref{sec:xform-join}, there are
several ways to eliminate a join operator. The right choice depends on whether
the join is one-to-one or one-to-many. To eliminate a join nest, the optimizer
performs an exhaustive search using the join elimination rules and uses the cost
model to choose the least expensive candidate.

The hash- and ordered-index introduction phases attempt to replace filter
operators with indexes. When replacing a filter operator with an index, the most
important choice to make is where in the query to place the filter. This choice
determines which part of the layout the index will partition. The optimizer
makes this choice by first hoisting all of the candidate filters as far as
possible. It then pushes the filters, introducing an index at each position. The
cost model is used to select the best candidate.

Finally the precomputation phase selects parts of the query that can be computed
and stored. Other transformation rules are interleaved with these phases.

The output of the optimizer is a sequence of transformation rules that lower the
input query to a layout, minimizing the cost of executing the resulting query. A
pleasant feature of the optimizer is that because it simply schedules
transformation rules, it is semantics preserving if all of the rules are. This
means that all schedules are equally correct---they differ only in the quality
of their optimization.

\stitle{Cost Model.} The staged nature of the layout algebra makes evaluating
the cost of a query complicated. We use the layout type (\Cref{sec:compiler}) to
estimate the cost of evaluating a query. The layout type tells us the size of
the collections in the layout, and we use simple models of the costs of the
runtime query operators to estimate the cost of the entire query. Computing the
layout type is expensive, so we use a sample of the database for cost modeling
during optimization.


\section{Evaluation}
We compare \projname{} with three other systems: \hyper{}~\cite{Neumann2011},
\cozy{}~\cite{Loncaric2018}, and \chestnut{}~\cite{Yan2019} (see \Cref{sec:related}).

\hyper{} is an in-memory column-store which has a state-of-the art query
compiler. It implements compilation techniques (e.g.\ vectorization) that are
well outside the scope of this paper. We compare against \hyper{} in two modes:
with the original TPC-H data and with a transformed version of the data and
query that mimics the layout used by \projname{}. We compare against vanilla
\hyper{} to show that layout specialization is a powerful optimization that can
compensate for the many low-level compiler optimizations in \hyper{}. We compare
against \hyper{} with transformed data to show that the specialization
techniques that \projname{} uses are beneficial in other systems.

In the comparison with \hyper{}, the \projname{} results are split into two
categories: expert generated queries and optimizer generated queries. In both
cases we start with a direct translation of the SQL implementation of each query
into the \langname{}. For the expert queries we hand-selected a sequence of
transformations that generates an efficient, serializable version of the query.
For the optimized queries, the optimizer (\Cref{sec:optimizer}) searches over
the space of transformation sequences, using its cost model to evaluate
candidates. In some cases, the optimizer fails to find a serializable candidate,
so there is no bar in the plot.

\cozy{} is a state-of-the-art program synthesis tool that generates specialized
data structures from relational queries. \chestnut{} is a tool for synthesizing
specialized data structures for object queries. We run both on TPC-H and compare
with \projname{}.

\subsection{TPC-H Analytics Benchmark}

TPC-H is a standard database benchmark, focusing on analytics queries. It
consists of a data generator, 22 query templates, and a query generator which
instantiates the templates. The queries in TPC-H are inherently parametric, and
their parameters come from the domains defined by the query generator. To build
our benchmark, we took the query templates from TPC-H and encoded them as
\projname{} programs. It is important that the queries be parametric.
Specializing non-parametric queries is uninteresting; a non-parametric query can
be evaluated and the result stored.

TPC-H is a general purpose benchmark, so it exercises a variety of SQL
primitives. We chose not to implement all of these primitives in \projname{},
not because they would be prohibitively difficult, but because they are not
directly related to the layout specialization problem. In particular,
\projname{} does not support executing \orderby{}, \groupby{}, \join{}, or
\dedup{} operators at runtime\footnote{These operators can be processed into the
compiled form of the query.}, and it does not support limit clauses at all.
Some of these operators can be replaced by layout specialization, but others
cannot. We implemented the first 19 queries in TPC-H. Of these queries, we
dropped query 13 because it contains an outer join and removed runtime ordering
and limit clauses from five other queries.

\subsection{Results}

\begin{figure}
  \includegraphics[width=\columnwidth]{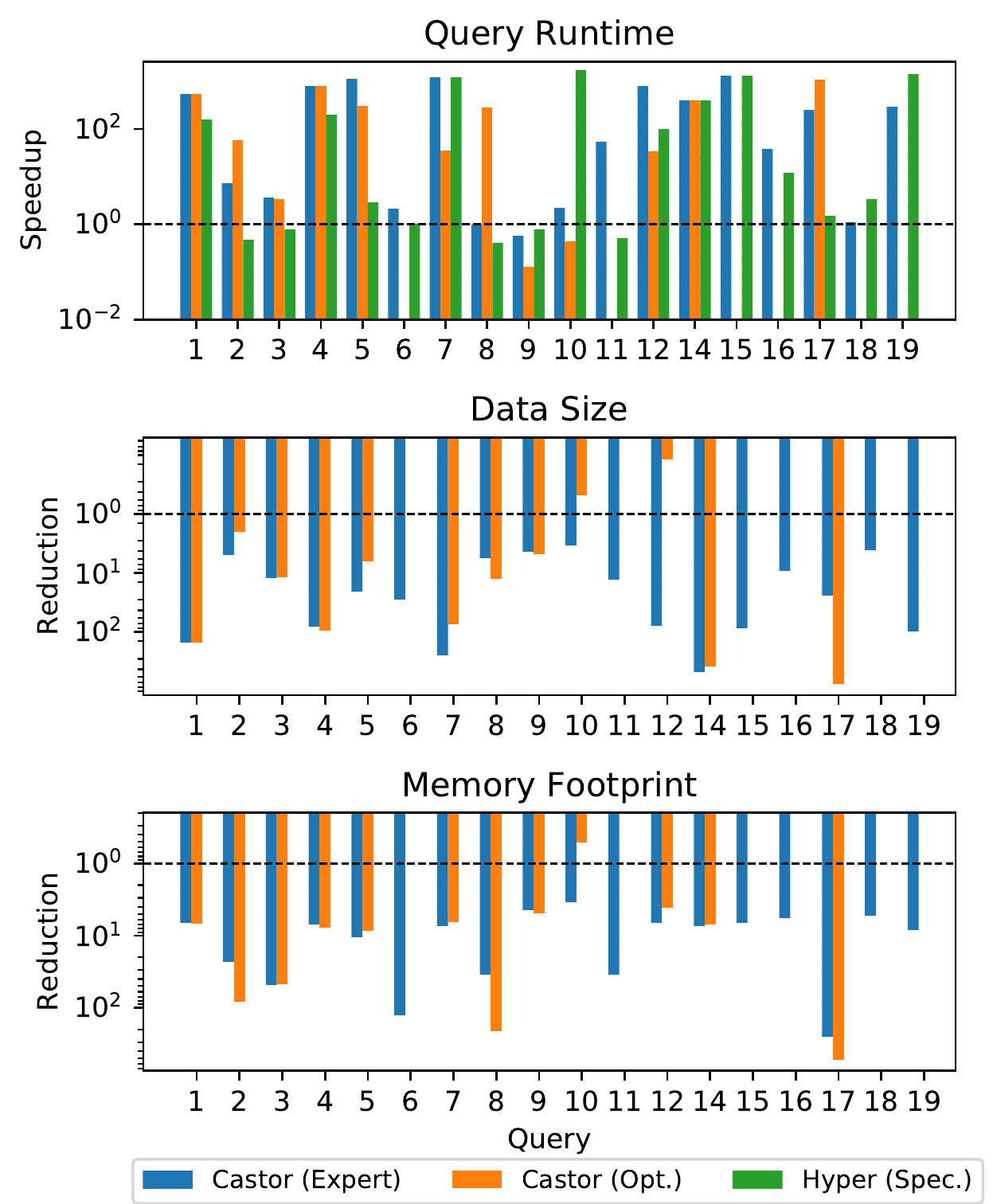}
  \caption{Performance on TPC-H queries.}\label{fig:perf}
  \vspace{-3ex}
\end{figure}


When evaluating the TPC-H queries, we used the 1Gb scale factor. We ran our
benchmarks on an Intel\textregistered{} Xeon\textregistered{} E5--2470 with 100Gb
of memory.

\stitle{Runtime.} \Cref{fig:perf} shows the speedup over baseline \hyper{}. The
query runtime numbers in \Cref{tbl:tpch} show that the layouts and query code
generated by \projname{} from expert queries are faster or significantly faster
than \hyper{} for 15 out of 18 queries. In the cases where \projname{} is slower
than \hyper{}, no query is more than 2x slower.

If \hyper{} is given specialized views and indexes, then its performance is on
par with \projname{}. However, constructing and maintaining these views takes
effort, and \hyper{} cannot assist the user in creating a collection of views
which maintains the semantics of the original query.

\stitle{Layout Size.} We recorded the size of the layouts that
\projname{} produced and compared them to the size of \hyper{}'s specialized
views and indexes. \Cref{fig:perf} shows that \projname{}'s specialized layouts
are much smaller than \hyper{}'s, even if \hyper{} uses the same
kind of data removal and indexing optimizations.

The layouts were generally small---less than 10Mb for 10 out of 18 queries and
less than 100Mb for all but one query. The original data set, the output of the
TPC-H data generator, is 1.1Gb. The size difference between \projname{}'s
layouts and the original data supports the hypothesis that queries, even
parameterized queries, rely on fairly small subsets of the whole database,
making layout specialization a profitable optimization.

\stitle{Memory Use.} We also measured the peak memory use of the query process
for \hyper{} and for \projname{}. \hyper{} consistently uses approximately the
same amount of memory as the layout size. In some cases it uses more, presumably
because it has large runtime dependencies like LLVM.\@ In contrast,
\Cref{fig:perf} shows that \projname{}'s peak memory use is significantly lower
than \hyper{} for all expert queries.

\stitle{\cozy{}.} We transformed our input queries into \cozy{}'s specification
format and ran \cozy{} with a 6 hour timeout. In this configuration, we found
that \cozy{} was unable to make significant improvements on all but two of the
TPC-H queries. In Q4, \cozy{} precomputed one of the joins and a filter. In Q17,
\cozy{} added an index. We ran both of these queries and found that despite the
optimization Q4 was slower than baseline \hyper{} at 5.4s and Q17 was too slow
to run on the entire TPC-H dataset. Although \cozy{} is effective at
synthesizing data structures from small relational specifications, the size of
the TPC-H queries causes a significant slowdown in its solver-based verification
step.

\stitle{\chestnut{}.} We attempted to use \chestnut{} to optimize four of the
TPC-H queries, but we were unable to build and run the generated code. Manually
examining the selected layouts for Q1 and Q3--6, we find that \chestnut{} uses
projection and indexes in many of the same places that \projname{} does, but
misses some optimizations that \projname{} can take advantage of, such as
aggregate precomputation (\Cref{sec:xform-precomp}).

\subsection{Optimizer}

The performance of the automatic optimizer is mixed. In some cases, it generates
queries that perform similarly or better than the expertly tuned queries. In
others, it fails to find a transformation sequence that generates a serializable
query, it generates a serializable query that uses a layout that is too large to
compile in a reasonable amount of time, or it generates a query which performs
poorly.

In most of the cases where the optimizer fails to find a solution, it fails
because the cost model provided an inaccurate estimate of the performance of the
candidate. In the cost model, we only consider the time required to run the
query, not the size of the layout. This avoids a complex multi-variable
optimization problem, but it means that the optimizer sometimes generates
layouts that are too large to be compiled efficiently.

Despite the mixed performance of the optimizer, it provides a good starting
point for an expert implementation, and it demonstrates that automatically
optimizing programs in the \langname{} is feasible.

\subsection{Summary}

We showed that \projname{} produces artifacts that are competitive in
performance and in size with a state-of-the-art in-memory database. These
results show that database compilation is a compelling technique for improving
query performance on static and slowly changing datasets.


\section{Related Work}\label{sec:related}

\stitle{Deductive Synthesis.} There is a long line of work that uses deductive
synthesis and program transformation rules to optimize
programs~\cite{Blaine1998,Puschel2005}, to generate data structure
implementations~\cite{Delaware2015}, and to build performance
DSLs~\cite{Ragan-Kelley2013,Sujeeth2014}. \projname{} is a part of this line of
work: it is a performance DSL which uses deduction rules to generate and
optimize layouts. However its focus on particular data sets and on deduction
rules to optimize data in addition to programs separates it from previous work.

\stitle{Data Representation Synthesis.} The layout optimization problem is
similar to the problem of synthesizing a data structure that corresponds to a
relational
specification~\cite{Hawkins2010,Hawkins2011,Loncaric2016,Loncaric2018,Sujeeth2014}.
\projname{} considers a restricted version of the data structure synthesis
problem where the query and the dataset are known to the compiler, which allows
\projname{} to use optimizations which would not be safe if the data was not
known. The best data structure synthesis tool---\cozy{}---uses an SMT solver to
verify candidates, which does not scale to the TPC-H queries. \projname{}'s use
of deduction rules avoids this costly verification step.

\stitle{Database Storage.} Traditional databases are mostly row-based.
Column-based database systems (e.g., \hyper{}~\cite{Neumann2011},
MonetDB~\cite{Boncz1999} and C-Store~\cite{Stonebraker2005a}) are popular for
OLAP applications, outperforming row-based approaches by orders of magnitude.
However, the existing work on database storage generally considers specific
storage optimizations (e.g., \cite{AilamakiDHS01}), rather than languages for
expressing diverse storage options. In this vein is
RodentStore~\cite{Cudre-Mauroux2009}, which proposed a language to express rich
types of storage layouts and showed that different layouts could benefit
different applications. However, a compiler was never developed to create the
layouts from this language; the paper demonstrated its point by implementing
each layout by hand. Also related is \chestnut{}: a tool for generating
specialized layouts for object queries~\cite{Yan2019}. \chestnut{} has separate
layout and query languages and synthesizes a query after choosing the layout.
This limits its ability to use transformations that change the data, like
predicate precomputation (\Cref{sec:xform-pred-precomp}).

There have also been studies of physical layouts for other types of data, such
as for scientific data~\cite{Stonebraker12}, and geo-spatial
data~\cite{GutierrezB07}. Although not directly comparable, we hope that
\projname{} can be extended to support those data types.

\stitle{Materialized View and Index Selection.} The layouts that \projname{}
generates are similar to materialized views, in that they store query results.
\projname{} also generates layouts which contain indexes. Several problems
related to the use of materialized views and indexes have been studied
(see~\cite{Halevy01} for a survey): (1) the view storage problem that decides
which views need to be materialized~\cite{ChirkovaG00}, (2) the view selection
problem that selects view(s) that can answer a given query, (3) the query
rewriting problem that rewrites the given query based on the selected
view(s)~\cite{PottingerL00}, (4) the index selection problem that selects an
appropriate set of indexes for a
query~\cite{Stonebraker74,Gupta1997,Bruno2005,Talebi2008}. However, materialized
views are restricted to being flat relations. The layout space that \projname{}
supports is much richer than that supported by materialized views and indexes.
In addition, the view selection literature has not previously considered the
problem of generating execution plans for chosen views and indexes.

\stitle{Query Compilation.} \projname{} uses techniques from the query
compilation literature~\cite{Tahboub2018,Shaikhha2016,Klonatos2014}. It extends
these techniques by using information about the layout to specialize its
queries.


\section{Conclusion}

We have presented \projname{}, a domain specific language for expressing a wide
variety of physical database designs, and a compiler for this language. We have
evaluated it empirically and shown that it is competitive with the
state-of-the-art in memory database systems.


\bibliography{citations}


\begin{thebibliography}{39}


\ifx \showCODEN    \undefined \def \showCODEN     #1{\unskip}     \fi
\ifx \showDOI      \undefined \def \showDOI       #1{#1}\fi
\ifx \showISBNx    \undefined \def \showISBNx     #1{\unskip}     \fi
\ifx \showISBNxiii \undefined \def \showISBNxiii  #1{\unskip}     \fi
\ifx \showISSN     \undefined \def \showISSN      #1{\unskip}     \fi
\ifx \showLCCN     \undefined \def \showLCCN      #1{\unskip}     \fi
\ifx \shownote     \undefined \def \shownote      #1{#1}          \fi
\ifx \showarticletitle \undefined \def \showarticletitle #1{#1}   \fi
\ifx \showURL      \undefined \def \showURL       {\relax}        \fi
\providecommand\bibfield[2]{#2}
\providecommand\bibinfo[2]{#2}
\providecommand\natexlab[1]{#1}
\providecommand\showeprint[2][]{arXiv:#2}

\bibitem[\protect\citeauthoryear{Ailamaki, DeWitt, Hill, and
  Skounakis}{Ailamaki et~al\mbox{.}}{2001}]%
        {AilamakiDHS01}
\bibfield{author}{\bibinfo{person}{Anastassia Ailamaki},
  \bibinfo{person}{David~J. DeWitt}, \bibinfo{person}{Mark~D. Hill}, {and}
  \bibinfo{person}{Marios Skounakis}.} \bibinfo{year}{2001}\natexlab{}.
\newblock \showarticletitle{Weaving Relations for Cache Performance}. In
  \bibinfo{booktitle}{\emph{VLDB}}. \bibinfo{pages}{169--180}.
\newblock


\bibitem[\protect\citeauthoryear{Blaine, Gilham, Liu, Smith, and
  Westfold}{Blaine et~al\mbox{.}}{1998}]%
        {Blaine1998}
\bibfield{author}{\bibinfo{person}{Lee Blaine}, \bibinfo{person}{Limei Gilham},
  \bibinfo{person}{Junbo Liu}, \bibinfo{person}{Douglas~R. Smith}, {and}
  \bibinfo{person}{Stephen Westfold}.} \bibinfo{year}{1998}\natexlab{}.
\newblock \showarticletitle{Planware-Domain-Specific Synthesis of
  High-Performance Schedulers}. In \bibinfo{booktitle}{\emph{Automated
  {{Software Engineering}}, 1998. {{Proceedings}}. 13th {{IEEE International
  Conference}} On}}. \bibinfo{publisher}{{IEEE}}, \bibinfo{pages}{270--279}.
\newblock


\bibitem[\protect\citeauthoryear{Boncz and Kersten}{Boncz and Kersten}{1999}]%
        {Boncz1999}
\bibfield{author}{\bibinfo{person}{Peter~A. Boncz} {and}
  \bibinfo{person}{Martin~L. Kersten}.} \bibinfo{year}{1999}\natexlab{}.
\newblock \showarticletitle{{{MIL Primitives}} for {{Querying}} a {{Fragmented
  World}}}.
\newblock \bibinfo{journal}{\emph{The VLDB Journal}} \bibinfo{volume}{8},
  \bibinfo{number}{2} (\bibinfo{date}{Oct.} \bibinfo{year}{1999}),
  \bibinfo{pages}{101--119}.
\newblock
\showISSN{1066-8888}
\urldef\tempurl%
\url{https://doi.org/10.1007/s007780050076}
\showDOI{\tempurl}


\bibitem[\protect\citeauthoryear{Botelho, Pagh, and Ziviani}{Botelho
  et~al\mbox{.}}{2007}]%
        {Botelho2007}
\bibfield{author}{\bibinfo{person}{Fabiano~C. Botelho}, \bibinfo{person}{Rasmus
  Pagh}, {and} \bibinfo{person}{Nivio Ziviani}.}
  \bibinfo{year}{2007}\natexlab{}.
\newblock \showarticletitle{Simple and Space-Efficient Minimal Perfect Hash
  Functions}. In \bibinfo{booktitle}{\emph{Algorithms and {{Data Structures}}:
  10th {{International Workshop}}, {{WADS}} 2007}}
  \emph{(\bibinfo{series}{Theoretical Computer Science and General Issues})},
  Vol.~\bibinfo{volume}{4619}. \bibinfo{publisher}{{Springer}},
  \bibinfo{address}{Halifax, Canada}, \bibinfo{pages}{139--150}.
\newblock


\bibitem[\protect\citeauthoryear{Bruno and Chaudhuri}{Bruno and
  Chaudhuri}{2005}]%
        {Bruno2005}
\bibfield{author}{\bibinfo{person}{Nicolas Bruno} {and}
  \bibinfo{person}{Surajit Chaudhuri}.} \bibinfo{year}{2005}\natexlab{}.
\newblock \showarticletitle{Automatic {{Physical Database Tuning}}: {{A
  Relaxation}}-Based {{Approach}}}. In \bibinfo{booktitle}{\emph{Proceedings of
  the 2005 {{ACM SIGMOD International Conference}} on {{Management}} of
  {{Data}}}} \emph{(\bibinfo{series}{SIGMOD '05})}. \bibinfo{publisher}{{ACM}},
  \bibinfo{address}{New York, NY, USA}, \bibinfo{pages}{227--238}.
\newblock
\showISBNx{978-1-59593-060-6}
\urldef\tempurl%
\url{https://doi.org/10.1145/1066157.1066184}
\showDOI{\tempurl}


\bibitem[\protect\citeauthoryear{Chaudhuri}{Chaudhuri}{1998}]%
        {Chaudhuri1998}
\bibfield{author}{\bibinfo{person}{Surajit Chaudhuri}.}
  \bibinfo{year}{1998}\natexlab{}.
\newblock \showarticletitle{An {{Overview}} of {{Query Optimization}} in
  {{Relational Systems}}}. In \bibinfo{booktitle}{\emph{Proceedings of the
  {{Seventeenth ACM SIGACT}}-{{SIGMOD}}-{{SIGART Symposium}} on {{Principles}}
  of {{Database Systems}}}} \emph{(\bibinfo{series}{PODS '98})}.
  \bibinfo{publisher}{{ACM}}, \bibinfo{address}{New York, NY, USA},
  \bibinfo{pages}{34--43}.
\newblock
\showISBNx{978-0-89791-996-8}
\urldef\tempurl%
\url{https://doi.org/10.1145/275487.275492}
\showDOI{\tempurl}


\bibitem[\protect\citeauthoryear{Cheung, {Solar-Lezama}, and Madden}{Cheung
  et~al\mbox{.}}{2013}]%
        {Cheung2013}
\bibfield{author}{\bibinfo{person}{Alvin Cheung}, \bibinfo{person}{Armando
  {Solar-Lezama}}, {and} \bibinfo{person}{Samuel Madden}.}
  \bibinfo{year}{2013}\natexlab{}.
\newblock \showarticletitle{Optimizing Database-Backed Applications with Query
  Synthesis}.
\newblock \bibinfo{journal}{\emph{ACM SIGPLAN Notices}} \bibinfo{volume}{48},
  \bibinfo{number}{6} (\bibinfo{year}{2013}), \bibinfo{pages}{3--14}.
\newblock
\urldef\tempurl%
\url{http://dl.acm.org/citation.cfm?id=2462180}
\showURL{%
\tempurl}


\bibitem[\protect\citeauthoryear{Chirkova and Genesereth}{Chirkova and
  Genesereth}{2000}]%
        {ChirkovaG00}
\bibfield{author}{\bibinfo{person}{Rada Chirkova} {and}
  \bibinfo{person}{Michael~R. Genesereth}.} \bibinfo{year}{2000}\natexlab{}.
\newblock \showarticletitle{Linearly Bounded Reformulations of Conjunctive
  Databases}. In \bibinfo{booktitle}{\emph{Computational Logic - {CL} 2000,
  First International Conference, London, UK, 24-28 July, 2000, Proceedings}}.
  \bibinfo{pages}{987--1001}.
\newblock
\urldef\tempurl%
\url{https://doi.org/10.1007/3-540-44957-4\_66}
\showDOI{\tempurl}


\bibitem[\protect\citeauthoryear{Codd}{Codd}{1970}]%
        {Codd1970}
\bibfield{author}{\bibinfo{person}{E.~F. Codd}.}
  \bibinfo{year}{1970}\natexlab{}.
\newblock \showarticletitle{A Relational Model of Data for Large Shared Data
  Banks}.
\newblock \bibinfo{journal}{\emph{Commun. ACM}} \bibinfo{volume}{13},
  \bibinfo{number}{6} (\bibinfo{date}{June} \bibinfo{year}{1970}),
  \bibinfo{pages}{377--387}.
\newblock
\showISSN{0001-0782}
\urldef\tempurl%
\url{https://doi.org/10.1145/362384.362685}
\showDOI{\tempurl}


\bibitem[\protect\citeauthoryear{Codd}{Codd}{1971}]%
        {Codd1971}
\bibfield{author}{\bibinfo{person}{Edgar~F. Codd}.}
  \bibinfo{year}{1971}\natexlab{}.
\newblock \showarticletitle{A Data Base Sublanguage Founded on the Relational
  Calculus}. In \bibinfo{booktitle}{\emph{Proceedings of the 1971 {{ACM
  SIGFIDET}} (Now {{SIGMOD}}) {{Workshop}} on {{Data Description}}, {{Access}}
  and {{Control}}}}. \bibinfo{publisher}{{ACM}}, \bibinfo{pages}{35--68}.
\newblock


\bibitem[\protect\citeauthoryear{Council}{Council}{2008}]%
        {Council2008}
\bibfield{author}{\bibinfo{person}{Transaction Processing~Performance
  Council}.} \bibinfo{year}{2008}\natexlab{}.
\newblock \showarticletitle{{{TPC}}-{{H}} Benchmark Specification}.
\newblock   \bibinfo{volume}{21} (\bibinfo{year}{2008}),
  \bibinfo{pages}{592--603}.
\newblock


\bibitem[\protect\citeauthoryear{{Cudre-Mauroux}, Wu, and
  Madden}{{Cudre-Mauroux} et~al\mbox{.}}{2009}]%
        {Cudre-Mauroux2009}
\bibfield{author}{\bibinfo{person}{Philippe {Cudre-Mauroux}},
  \bibinfo{person}{Eugene Wu}, {and} \bibinfo{person}{Sam Madden}.}
  \bibinfo{year}{2009}\natexlab{}.
\newblock \showarticletitle{The {{Case}} for {{RodentStore}}, an {{Adaptive}},
  {{Declarative Storage System}}}. In \bibinfo{booktitle}{\emph{{{CIDR}}}}.
\newblock
\showeprint[arxiv]{0909.1779}
\urldef\tempurl%
\url{http://arxiv.org/abs/0909.1779}
\showURL{%
\tempurl}


\bibitem[\protect\citeauthoryear{{Davi de Castro Reis}, {Djamel Belazzougui},
  {Fabiano Cupertino Botelho}, and {Nivio Ziviani}}{{Davi de Castro Reis}
  et~al\mbox{.}}{2011}]%
        {DavideCastroReis2011}
\bibfield{author}{\bibinfo{person}{{Davi de Castro Reis}},
  \bibinfo{person}{{Djamel Belazzougui}}, \bibinfo{person}{{Fabiano Cupertino
  Botelho}}, {and} \bibinfo{person}{{Nivio Ziviani}}.}
  \bibinfo{year}{2011}\natexlab{}.
\newblock \bibinfo{title}{{{CMPH}}: {{C}} Minimal Perfect Hashing Library}.
\newblock
\newblock
\urldef\tempurl%
\url{http://cmph.sourceforge.net}
\showURL{%
\tempurl}


\bibitem[\protect\citeauthoryear{Delaware, {Pit-Claudel}, Gross, and
  Chlipala}{Delaware et~al\mbox{.}}{2015}]%
        {Delaware2015}
\bibfield{author}{\bibinfo{person}{Benjamin Delaware},
  \bibinfo{person}{Cl\'ement {Pit-Claudel}}, \bibinfo{person}{Jason Gross},
  {and} \bibinfo{person}{Adam Chlipala}.} \bibinfo{year}{2015}\natexlab{}.
\newblock \showarticletitle{Fiat: {{Deductive}} Synthesis of Abstract Data
  Types in a Proof Assistant}. In \bibinfo{booktitle}{\emph{{{ACM SIGPLAN
  Notices}}}}, Vol.~\bibinfo{volume}{50}. \bibinfo{publisher}{{ACM}},
  \bibinfo{pages}{689--700}.
\newblock


\bibitem[\protect\citeauthoryear{Graefe}{Graefe}{1994}]%
        {Graefe1994}
\bibfield{author}{\bibinfo{person}{Goetz Graefe}.}
  \bibinfo{year}{1994}\natexlab{}.
\newblock \showarticletitle{Volcano/Spl Minus/an Extensible and Parallel Query
  Evaluation System}.
\newblock \bibinfo{journal}{\emph{IEEE Transactions on Knowledge and Data
  Engineering}} \bibinfo{volume}{6}, \bibinfo{number}{1}
  (\bibinfo{year}{1994}), \bibinfo{pages}{120--135}.
\newblock


\bibitem[\protect\citeauthoryear{Gupta, Harinarayan, Rajaraman, and
  Ullman}{Gupta et~al\mbox{.}}{1997}]%
        {Gupta1997}
\bibfield{author}{\bibinfo{person}{H. Gupta}, \bibinfo{person}{V. Harinarayan},
  \bibinfo{person}{A. Rajaraman}, {and} \bibinfo{person}{J.~D. Ullman}.}
  \bibinfo{year}{1997}\natexlab{}.
\newblock \showarticletitle{Index Selection for {{OLAP}}}. In
  \bibinfo{booktitle}{\emph{Proceedings 13th {{International Conference}} on
  {{Data Engineering}}}}. \bibinfo{publisher}{{IEEE Computer Society}},
  \bibinfo{address}{Junglee Corp., Palo Alto, CA.}, \bibinfo{pages}{208--219}.
\newblock
\urldef\tempurl%
\url{https://doi.org/10.1109/ICDE.1997.581755}
\showDOI{\tempurl}


\bibitem[\protect\citeauthoryear{Guti{\'{e}}rrez and Baumann}{Guti{\'{e}}rrez
  and Baumann}{2007}]%
        {GutierrezB07}
\bibfield{author}{\bibinfo{person}{Ang{\'{e}}lica~Garc{\'{\i}}a
  Guti{\'{e}}rrez} {and} \bibinfo{person}{Peter Baumann}.}
  \bibinfo{year}{2007}\natexlab{}.
\newblock \showarticletitle{Modeling Fundamental Geo-Raster Operations with
  Array Algebra}. In \bibinfo{booktitle}{\emph{Workshops Proceedings of the 7th
  {IEEE} International Conference on Data Mining {(ICDM} 2007), October 28-31,
  2007, Omaha, Nebraska, {USA}}}. \bibinfo{pages}{607--612}.
\newblock
\urldef\tempurl%
\url{https://doi.org/10.1109/ICDMW.2007.53}
\showDOI{\tempurl}


\bibitem[\protect\citeauthoryear{Halevy}{Halevy}{2001}]%
        {Halevy01}
\bibfield{author}{\bibinfo{person}{Alon~Y. Halevy}.}
  \bibinfo{year}{2001}\natexlab{}.
\newblock \showarticletitle{Answering queries using views: {A} survey}.
\newblock \bibinfo{journal}{\emph{{VLDB} J.}} \bibinfo{volume}{10},
  \bibinfo{number}{4} (\bibinfo{year}{2001}), \bibinfo{pages}{270--294}.
\newblock
\urldef\tempurl%
\url{https://doi.org/10.1007/s007780100054}
\showDOI{\tempurl}


\bibitem[\protect\citeauthoryear{Hawkins, Aiken, Fisher, Rinard, and
  Sagiv}{Hawkins et~al\mbox{.}}{2010}]%
        {Hawkins2010}
\bibfield{author}{\bibinfo{person}{Peter Hawkins}, \bibinfo{person}{Alex
  Aiken}, \bibinfo{person}{Kathleen Fisher}, \bibinfo{person}{Martin Rinard},
  {and} \bibinfo{person}{Mooly Sagiv}.} \bibinfo{year}{2010}\natexlab{}.
\newblock \showarticletitle{Data {{Structure Fusion}}}. In
  \bibinfo{booktitle}{\emph{Programming {{Languages}} and {{Systems}}}}
  \emph{(\bibinfo{series}{Lecture Notes in Computer Science})}.
  \bibinfo{publisher}{{Springer, Berlin, Heidelberg}},
  \bibinfo{pages}{204--221}.
\newblock
\showISBNx{978-3-642-17163-5 978-3-642-17164-2}
\urldef\tempurl%
\url{https://doi.org/10.1007/978-3-642-17164-2_15}
\showDOI{\tempurl}


\bibitem[\protect\citeauthoryear{Hawkins, Aiken, Fisher, Rinard, and
  Sagiv}{Hawkins et~al\mbox{.}}{2011}]%
        {Hawkins2011}
\bibfield{author}{\bibinfo{person}{Peter Hawkins}, \bibinfo{person}{Alex
  Aiken}, \bibinfo{person}{Kathleen Fisher}, \bibinfo{person}{Martin Rinard},
  {and} \bibinfo{person}{Mooly Sagiv}.} \bibinfo{year}{2011}\natexlab{}.
\newblock \showarticletitle{Data {{Representation Synthesis}}}. In
  \bibinfo{booktitle}{\emph{Proceedings of the {{32Nd ACM SIGPLAN Conference}}
  on {{Programming Language Design}} and {{Implementation}}}}
  \emph{(\bibinfo{series}{PLDI '11})}. \bibinfo{publisher}{{ACM}},
  \bibinfo{address}{New York, NY, USA}, \bibinfo{pages}{38--49}.
\newblock
\showISBNx{978-1-4503-0663-8}
\urldef\tempurl%
\url{https://doi.org/10.1145/1993498.1993504}
\showDOI{\tempurl}


\bibitem[\protect\citeauthoryear{Jarke and Koch}{Jarke and Koch}{1984}]%
        {Jarke1984}
\bibfield{author}{\bibinfo{person}{Matthias Jarke} {and}
  \bibinfo{person}{Jurgen Koch}.} \bibinfo{year}{1984}\natexlab{}.
\newblock \showarticletitle{Query Optimization in Database Systems}.
\newblock \bibinfo{journal}{\emph{Comput. Surveys}} \bibinfo{volume}{16},
  \bibinfo{number}{2} (\bibinfo{date}{June} \bibinfo{year}{1984}),
  \bibinfo{pages}{111--152}.
\newblock
\showISSN{0360-0300}
\urldef\tempurl%
\url{https://doi.org/10.1145/356924.356928}
\showDOI{\tempurl}


\bibitem[\protect\citeauthoryear{Klonatos, Koch, Rompf, and Chafi}{Klonatos
  et~al\mbox{.}}{2014}]%
        {Klonatos2014}
\bibfield{author}{\bibinfo{person}{Yannis Klonatos}, \bibinfo{person}{Christoph
  Koch}, \bibinfo{person}{Tiark Rompf}, {and} \bibinfo{person}{Hassan Chafi}.}
  \bibinfo{year}{2014}\natexlab{}.
\newblock \showarticletitle{Building Efficient Query Engines in a High-Level
  Language}.
\newblock \bibinfo{journal}{\emph{Proceedings of the VLDB Endowment}}
  \bibinfo{volume}{7}, \bibinfo{number}{10} (\bibinfo{year}{2014}),
  \bibinfo{pages}{853--864}.
\newblock
\urldef\tempurl%
\url{http://dl.acm.org/citation.cfm?id=2732959}
\showURL{%
\tempurl}


\bibitem[\protect\citeauthoryear{Loncaric, Ernst, and Torlak}{Loncaric
  et~al\mbox{.}}{2018}]%
        {Loncaric2018}
\bibfield{author}{\bibinfo{person}{Calvin Loncaric},
  \bibinfo{person}{Michael~D. Ernst}, {and} \bibinfo{person}{Emina Torlak}.}
  \bibinfo{year}{2018}\natexlab{}.
\newblock \showarticletitle{Generalized Data Structure Synthesis}. In
  \bibinfo{booktitle}{\emph{Proceedings of the 40th {{International
  Conference}} on {{Software Engineering}}}}. \bibinfo{publisher}{{ACM}},
  \bibinfo{pages}{958--968}.
\newblock


\bibitem[\protect\citeauthoryear{Loncaric, Torlak, and Ernst}{Loncaric
  et~al\mbox{.}}{2016}]%
        {Loncaric2016}
\bibfield{author}{\bibinfo{person}{Calvin Loncaric}, \bibinfo{person}{Emina
  Torlak}, {and} \bibinfo{person}{Michael~D. Ernst}.}
  \bibinfo{year}{2016}\natexlab{}.
\newblock \showarticletitle{Fast {{Synthesis}} of {{Fast Collections}}}. In
  \bibinfo{booktitle}{\emph{Proceedings of the 37th {{ACM SIGPLAN Conference}}
  on {{Programming Language Design}} and {{Implementation}}}}
  \emph{(\bibinfo{series}{PLDI '16})}. \bibinfo{publisher}{{ACM}},
  \bibinfo{address}{New York, NY, USA}, \bibinfo{pages}{355--368}.
\newblock
\showISBNx{978-1-4503-4261-2}
\urldef\tempurl%
\url{https://doi.org/10.1145/2908080.2908122}
\showDOI{\tempurl}


\bibitem[\protect\citeauthoryear{Neumann}{Neumann}{2011}]%
        {Neumann2011}
\bibfield{author}{\bibinfo{person}{Thomas Neumann}.}
  \bibinfo{year}{2011}\natexlab{}.
\newblock \showarticletitle{Efficiently {{Compiling Efficient Query Plans}} for
  {{Modern Hardware}}}.
\newblock \bibinfo{journal}{\emph{Proc. VLDB Endow.}} \bibinfo{volume}{4},
  \bibinfo{number}{9} (\bibinfo{date}{June} \bibinfo{year}{2011}),
  \bibinfo{pages}{539--550}.
\newblock
\showISSN{2150-8097}
\urldef\tempurl%
\url{https://doi.org/10.14778/2002938.2002940}
\showDOI{\tempurl}


\bibitem[\protect\citeauthoryear{Pottinger and Levy}{Pottinger and
  Levy}{2000}]%
        {PottingerL00}
\bibfield{author}{\bibinfo{person}{Rachel Pottinger} {and}
  \bibinfo{person}{Alon~Y. Levy}.} \bibinfo{year}{2000}\natexlab{}.
\newblock \showarticletitle{A Scalable Algorithm for Answering Queries Using
  Views}. In \bibinfo{booktitle}{\emph{{VLDB} 2000, Proceedings of 26th
  International Conference on Very Large Data Bases, September 10-14, 2000,
  Cairo, Egypt}}. \bibinfo{pages}{484--495}.
\newblock
\urldef\tempurl%
\url{http://www.vldb.org/conf/2000/P484.pdf}
\showURL{%
\tempurl}


\bibitem[\protect\citeauthoryear{Puschel, Moura, Johnson, Padua, Veloso,
  Singer, Xiong, Franchetti, Gacic, Voronenko, Chen, Johnson, and
  Rizzolo}{Puschel et~al\mbox{.}}{2005}]%
        {Puschel2005}
\bibfield{author}{\bibinfo{person}{M. Puschel}, \bibinfo{person}{J.~M.~F.
  Moura}, \bibinfo{person}{J.~R. Johnson}, \bibinfo{person}{D. Padua},
  \bibinfo{person}{M.~M. Veloso}, \bibinfo{person}{B.~W. Singer},
  \bibinfo{person}{Jianxin Xiong}, \bibinfo{person}{F. Franchetti},
  \bibinfo{person}{A. Gacic}, \bibinfo{person}{Y. Voronenko},
  \bibinfo{person}{K. Chen}, \bibinfo{person}{R.~W. Johnson}, {and}
  \bibinfo{person}{N. Rizzolo}.} \bibinfo{year}{2005}\natexlab{}.
\newblock \showarticletitle{{{SPIRAL}}: {{Code Generation}} for {{DSP
  Transforms}}}.
\newblock \bibinfo{journal}{\emph{Proc. IEEE}} \bibinfo{volume}{93},
  \bibinfo{number}{2} (\bibinfo{date}{Feb.} \bibinfo{year}{2005}),
  \bibinfo{pages}{232--275}.
\newblock
\showISSN{0018-9219}
\urldef\tempurl%
\url{https://doi.org/10.1109/JPROC.2004.840306}
\showDOI{\tempurl}


\bibitem[\protect\citeauthoryear{{Ragan-Kelley}, Barnes, Adams, Paris, Durand,
  and Amarasinghe}{{Ragan-Kelley} et~al\mbox{.}}{2013}]%
        {Ragan-Kelley2013}
\bibfield{author}{\bibinfo{person}{Jonathan {Ragan-Kelley}},
  \bibinfo{person}{Connelly Barnes}, \bibinfo{person}{Andrew Adams},
  \bibinfo{person}{Sylvain Paris}, \bibinfo{person}{Fr\'edo Durand}, {and}
  \bibinfo{person}{Saman Amarasinghe}.} \bibinfo{year}{2013}\natexlab{}.
\newblock \showarticletitle{Halide: A Language and Compiler for Optimizing
  Parallelism, Locality, and Recomputation in Image Processing Pipelines}.
\newblock \bibinfo{journal}{\emph{ACM SIGPLAN Notices}} \bibinfo{volume}{48},
  \bibinfo{number}{6} (\bibinfo{year}{2013}), \bibinfo{pages}{519--530}.
\newblock


\bibitem[\protect\citeauthoryear{Rompf and Amin}{Rompf and Amin}{2015}]%
        {Rompf2015}
\bibfield{author}{\bibinfo{person}{Tiark Rompf} {and} \bibinfo{person}{Nada
  Amin}.} \bibinfo{year}{2015}\natexlab{}.
\newblock \showarticletitle{Functional {{Pearl}}: {{A SQL}} to {{C Compiler}}
  in 500 {{Lines}} of {{Code}}}. In \bibinfo{booktitle}{\emph{Proceedings of
  the 20th {{ACM SIGPLAN International Conference}} on {{Functional
  Programming}}}} \emph{(\bibinfo{series}{ICFP 2015})}.
  \bibinfo{publisher}{{ACM}}, \bibinfo{address}{New York, NY, USA},
  \bibinfo{pages}{2--9}.
\newblock
\showISBNx{978-1-4503-3669-7}
\urldef\tempurl%
\url{https://doi.org/10.1145/2784731.2784760}
\showDOI{\tempurl}


\bibitem[\protect\citeauthoryear{Shaikhha, Klonatos, Parreaux, Brown, Dashti,
  and Koch}{Shaikhha et~al\mbox{.}}{2016}]%
        {Shaikhha2016}
\bibfield{author}{\bibinfo{person}{Amir Shaikhha}, \bibinfo{person}{Yannis
  Klonatos}, \bibinfo{person}{Lionel Parreaux}, \bibinfo{person}{Lewis Brown},
  \bibinfo{person}{Mohammad Dashti}, {and} \bibinfo{person}{Christoph Koch}.}
  \bibinfo{year}{2016}\natexlab{}.
\newblock \showarticletitle{How to {{Architect}} a {{Query Compiler}}}. In
  \bibinfo{booktitle}{\emph{Proceedings of the 2016 {{International
  Conference}} on {{Management}} of {{Data}}}} \emph{(\bibinfo{series}{SIGMOD
  '16})}. \bibinfo{publisher}{{ACM}}, \bibinfo{address}{New York, NY, USA},
  \bibinfo{pages}{1907--1922}.
\newblock
\showISBNx{978-1-4503-3531-7}
\urldef\tempurl%
\url{https://doi.org/10.1145/2882903.2915244}
\showDOI{\tempurl}


\bibitem[\protect\citeauthoryear{Stonebraker}{Stonebraker}{1974}]%
        {Stonebraker74}
\bibfield{author}{\bibinfo{person}{Michael Stonebraker}.}
  \bibinfo{year}{1974}\natexlab{}.
\newblock \showarticletitle{The choice of partial inversions and combined
  indices}.
\newblock \bibinfo{journal}{\emph{International Journal of Parallel
  Programming}} \bibinfo{volume}{3}, \bibinfo{number}{2}
  (\bibinfo{year}{1974}), \bibinfo{pages}{167--188}.
\newblock
\urldef\tempurl%
\url{https://doi.org/10.1007/BF00976642}
\showDOI{\tempurl}


\bibitem[\protect\citeauthoryear{Stonebraker}{Stonebraker}{2012}]%
        {Stonebraker12}
\bibfield{author}{\bibinfo{person}{Michael Stonebraker}.}
  \bibinfo{year}{2012}\natexlab{}.
\newblock \showarticletitle{SciDB: An Open-Source {DBMS} for Scientific Data}.
\newblock \bibinfo{journal}{\emph{{ERCIM} News}} \bibinfo{volume}{2012},
  \bibinfo{number}{89} (\bibinfo{year}{2012}).
\newblock
\urldef\tempurl%
\url{http://ercim-news.ercim.eu/en89/special/scidb-an-open-source-dbms-for-scientific-data}
\showURL{%
\tempurl}


\bibitem[\protect\citeauthoryear{Stonebraker, Abadi, Batkin, Chen, Cherniack,
  Ferreira, Lau, Lin, Madden, O'Neil, and {others}}{Stonebraker
  et~al\mbox{.}}{2005}]%
        {Stonebraker2005a}
\bibfield{author}{\bibinfo{person}{Mike Stonebraker},
  \bibinfo{person}{Daniel~J. Abadi}, \bibinfo{person}{Adam Batkin},
  \bibinfo{person}{Xuedong Chen}, \bibinfo{person}{Mitch Cherniack},
  \bibinfo{person}{Miguel Ferreira}, \bibinfo{person}{Edmond Lau},
  \bibinfo{person}{Amerson Lin}, \bibinfo{person}{Sam Madden},
  \bibinfo{person}{Elizabeth O'Neil}, {and} \bibinfo{person}{{others}}.}
  \bibinfo{year}{2005}\natexlab{}.
\newblock \showarticletitle{C-Store: A Column-Oriented {{DBMS}}}. In
  \bibinfo{booktitle}{\emph{Proceedings of the 31st International Conference on
  {{Very}} Large Data Bases}}. \bibinfo{publisher}{{VLDB Endowment}},
  \bibinfo{pages}{553--564}.
\newblock
\urldef\tempurl%
\url{http://dl.acm.org/citation.cfm?id=1083658}
\showURL{%
\tempurl}


\bibitem[\protect\citeauthoryear{Sujeeth, Brown, Lee, Rompf, Chafi, Odersky,
  and Olukotun}{Sujeeth et~al\mbox{.}}{2014}]%
        {Sujeeth2014}
\bibfield{author}{\bibinfo{person}{Arvind~K. Sujeeth},
  \bibinfo{person}{Kevin~J. Brown}, \bibinfo{person}{Hyoukjoong Lee},
  \bibinfo{person}{Tiark Rompf}, \bibinfo{person}{Hassan Chafi},
  \bibinfo{person}{Martin Odersky}, {and} \bibinfo{person}{Kunle Olukotun}.}
  \bibinfo{year}{2014}\natexlab{}.
\newblock \showarticletitle{Delite: {{A Compiler Architecture}} for
  {{Performance}}-{{Oriented Embedded Domain}}-{{Specific Languages}}}.
\newblock \bibinfo{journal}{\emph{ACM Trans. Embed. Comput. Syst.}}
  \bibinfo{volume}{13}, \bibinfo{number}{4s} (\bibinfo{date}{April}
  \bibinfo{year}{2014}), \bibinfo{pages}{134:1--134:25}.
\newblock
\showISSN{1539-9087}
\urldef\tempurl%
\url{https://doi.org/10.1145/2584665}
\showDOI{\tempurl}


\bibitem[\protect\citeauthoryear{Tahboub, Essertel, and Rompf}{Tahboub
  et~al\mbox{.}}{2018}]%
        {Tahboub2018}
\bibfield{author}{\bibinfo{person}{Ruby~Y. Tahboub},
  \bibinfo{person}{Gr\'egory~M. Essertel}, {and} \bibinfo{person}{Tiark
  Rompf}.} \bibinfo{year}{2018}\natexlab{}.
\newblock \showarticletitle{How to {{Architect}} a {{Query Compiler}},
  {{Revisited}}}. In \bibinfo{booktitle}{\emph{Proceedings of the 2018
  {{International Conference}} on {{Management}} of {{Data}}}}.
  \bibinfo{publisher}{{ACM}}, \bibinfo{pages}{307--322}.
\newblock


\bibitem[\protect\citeauthoryear{Talebi, Chirkova, Fathi, and Stallmann}{Talebi
  et~al\mbox{.}}{2008}]%
        {Talebi2008}
\bibfield{author}{\bibinfo{person}{Zohreh~Asgharzadeh Talebi},
  \bibinfo{person}{Rada Chirkova}, \bibinfo{person}{Yahya Fathi}, {and}
  \bibinfo{person}{Matthias Stallmann}.} \bibinfo{year}{2008}\natexlab{}.
\newblock \showarticletitle{Exact and Inexact Methods for Selecting Views and
  Indexes for {{OLAP}} Performance Improvement}. In
  \bibinfo{booktitle}{\emph{Proceedings of the 11th International Conference on
  {{Extending}} Database Technology: {{Advances}} in Database Technology}}.
  \bibinfo{publisher}{{ACM}}, \bibinfo{pages}{311--322}.
\newblock
\urldef\tempurl%
\url{http://dl.acm.org/citation.cfm?id=1353383}
\showURL{%
\tempurl}


\bibitem[\protect\citeauthoryear{Visser}{Visser}{2005}]%
        {Visser2005}
\bibfield{author}{\bibinfo{person}{Eelco Visser}.}
  \bibinfo{year}{2005}\natexlab{}.
\newblock \showarticletitle{A Survey of Strategies in Rule-Based Program
  Transformation Systems}.
\newblock \bibinfo{journal}{\emph{Journal of Symbolic Computation}}
  \bibinfo{volume}{40}, \bibinfo{number}{1} (\bibinfo{date}{July}
  \bibinfo{year}{2005}), \bibinfo{pages}{831--873}.
\newblock
\showISSN{0747-7171}
\urldef\tempurl%
\url{https://doi.org/10.1016/j.jsc.2004.12.011}
\showDOI{\tempurl}


\bibitem[\protect\citeauthoryear{Yan and Cheung}{Yan and Cheung}{2019}]%
        {Yan2019}
\bibfield{author}{\bibinfo{person}{Cong Yan} {and} \bibinfo{person}{Alvin
  Cheung}.} \bibinfo{year}{2019}\natexlab{}.
\newblock \showarticletitle{Generating Application-Specific Data Layouts for
  in-Memory Databases}.
\newblock \bibinfo{journal}{\emph{Proceedings of the VLDB Endowment}}
  \bibinfo{volume}{12}, \bibinfo{number}{11} (\bibinfo{year}{2019}),
  \bibinfo{pages}{1513--1525}.
\newblock


\bibitem[\protect\citeauthoryear{Yessenov, Kuraj, and {Solar-Lezama}}{Yessenov
  et~al\mbox{.}}{2017}]%
        {Yessenov2017}
\bibfield{author}{\bibinfo{person}{Kuat Yessenov}, \bibinfo{person}{Ivan
  Kuraj}, {and} \bibinfo{person}{Armando {Solar-Lezama}}.}
  \bibinfo{year}{2017}\natexlab{}.
\newblock \showarticletitle{{{DemoMatch}}: {{API}} Discovery from
  Demonstrations}. In \bibinfo{booktitle}{\emph{{{PLDI}}}}.
  \bibinfo{publisher}{{ACM}}, \bibinfo{address}{Barcelona, Spain},
  \bibinfo{pages}{15}.
\newblock
\urldef\tempurl%
\url{https://doi.org/10.1145/3062341.3062386}
\showDOI{\tempurl}


\end{thebibliography}


\clearpage
\appendix
\section{Layout Semantics}\label{sec:layout-semantics}

In this section we discuss the layout semantics, which specifies how a layout
algebra program may be serialized to a binary format.

Each of the layout operators has a serialization format that is designed to be
as compact as possible. These are as follows:
\begin{itemize}
\item Integers are stored using the minimum number of bytes, from 1 to 8 bytes.
\item Booleans are stored as single bytes.
\item Fixed point numbers are normalized to a fixed scale, and stored as
  integers.
\item Tuples are stored as the concatenation of the layouts they contain,
  prefixed by a length.
\item Lists are stored as a length followed by the concatenation of their
  elements. They can be efficiently scanned through, but not accessed randomly
  by index.
\item Hash indexes are implemented using minimal perfect
  hashes~\cite{Botelho2007,DavideCastroReis2011}. The hash values are stored as
  in a list, but during serialization a lookup table is generated using the CMPH
  library and stored before the values. Using perfect hashing allows the hash
  indexes to have load factors up to 99\%.
\item Ordered indexes are similar to hash indexes in that they store a lookup
  table in addition to storing the values. In the case of the ordered index keys
  are stored sorted and the correct range is found by binary search.
\end{itemize}

Serialization proceeds in two passes. First, we compute a layout type according
to the rules in \Cref{fig:layout-types}. This type is an abstraction of the
layout; we use it to specialize the layout to the data that it stores. For
example, we use an interval abstraction to represent integer scalars as well as
lengths of collections like lists. We use these intervals to choose the number
of bytes to use for these integers when serializing the layout. In some cases we
are able to use this abstraction to avoid storing anything at all. For example,
if we know that the length field of a tuple is always the same, we can avoid
storing that field and instead bake it into the layout reading code. This is a
surprisingly important optimization; for a tuple containing two integers, a
naive implementation would spend a third of the tuple's bytes just to store the
length field.

After computing the layout type, we serialize the layout according to the rules
in \Cref{fig:layout-semantics}.

\begin{figure}
  {\small
    \begin{alignat*}{3}
      n &::=&&\ \mathbb{Z} \quad r ::=\ [n, n] \\
      t &::=&&\ \textsf{intT}(r) ~|~ \textsf{boolT} ~|~ \textsf{fixedT}(r,
      n_{scale}) ~|~ \textsf{stringT}(r_{chars}) \\
      &\ccol{|}&&\ \textsf{tupleT}([t_1, \dots, t_k]) ~|~ \textsf{listT}(t, n_{elems}) ~|~\textsf{hash-idxT}(t_k, t_v) \\
      &\ccol{|}&&\  \textsf{ordered-idxT}(t_k, t_v) ~|~ \textsf{emptyT}
    \end{alignat*}
  } 
  \vspace*{-4ex}
  \caption{Syntax of the layout types.}\label{fig:layout-types}
\end{figure}


\begin{figure}
{\small
  \begin{gather*}
    \begin{aligned}
    \inference{
      \sigma |- \scalar{e \mapsto n} \evalto x \\ x\ \text{is an integer} 
    }{
      \sigma |- \scalar{e \mapsto n} : \textsf{intT}([x, x])
    } &
    \inference{
      \sigma |- q_1 : t_1, \dots, \sigma |- q_k : t_k \\ t = \textsf{tupleT}([t_1,
      \dots, t_k], \tau)
    }{
      \sigma |- \atuple{\tau}([q_1, \dots, q_k]) : t
    }
  \end{aligned} \\
  \begin{aligned}
    \inference{
     \sigma |- q_k \evalto r \\ t_k = \bigsqcup_{\sigma' \in r, \sigma' |- q_v :
       t'} t' \\ t = \textsf{listT}(t, [ |r|, |r| ])
    }{
      \alist{}(\as{q_k}{n}, q_v) : t
    } &
    \inference{
      t_1 = \textsf{intT}([q_1, h_1]) & t_2 = \textsf{intT}([q_2, h_2])
    }{
      t_1 \sqcup t_2 = \textsf{intT}([\min (q_1, q_2), \max (h_1, h_2)])
    }
  \end{aligned}
\end{gather*}
} 
 \vspace*{-3ex}  
  \caption{Selected semantics of the type inference pass.}\label{fig:layout-infer}
\vspace*{-3ex}    
\end{figure}


\begin{figure}
{\small
  \begin{gather*}
    b : Byte\ string\quad\sigma, t : Tuple\quad\delta : Id \mapsto Relation\\
    \inference{}{\sigma, \delta |- \emptyset \layoutto ""}\quad
    \inference{
      \sigma, \delta |- e \evalto v & b\ \text{is the binary format of}\ v
    }{
      \sigma, \delta |- \scalar{e} \layoutto b
    }\\[2ex]
    \inference{
      \sigma, \delta |- q_k \evalto [t_1, \dots, t_n] & \forall 1\leq i \leq n.\
      \sigma \cup t_i, \delta |- q_v \layoutto b_i \\
      \sigma, \delta |- \scalar{|b_1| + \cdots + |b_n|} \layoutto b_{len} & \sigma, \delta |- \scalar{n} \layoutto b_{ct} 
    }{
      \sigma, \delta |- \alist{}(q_k, q_v) \layoutto b_{ct}b_{len}b_1\dots b_n 
    }\\[2ex]
    \inference{
      \forall 1 \leq i \leq n.\ \sigma, \delta |- q_i \layoutto b_i &
      \sigma, \delta |- \scalar{|b_1| + \cdots + |b_n|} \layoutto b_{len} 
    }{
      \sigma, \delta |- \atuple{\tau}([q_1, \dots, q_n]) \layoutto b_{len}b_1\dots b_n
    }\\[2ex]
    \inference{
      \sigma, \delta |- q \layoutto b 
    }{
      \sigma, \delta |- \filter{}(e, q) \layoutto b
    }\quad
    \inference{
      \sigma, \delta |- q \layoutto b 
    }{
      \sigma, \delta |- \select{}(e, q) \layoutto b
    }
  \end{gather*}
} 
 \vspace*{-3ex}    
  \caption{Selected semantics of the layout serialization pass.}\label{fig:layout-semantics}
\end{figure}



\section{Runtime Semantics}\label{sec:runtime-semantics}

In this section we describe how layout algebra programs are compiled to
executable code and how that compilation process uses the layout type.

As mentioned in \Cref{sec:compiler}, query code is generated as
in~\cite{Tahboub2018}. This method is referred to as push-based, or data-centric
query evaluation. For each query operator, the code generator contains a
function that emits the code that implements the operator. Rather than emitting
an iterator which can be stepped forward at runtime, these functions take a
callback which generates the code that consumes the output of the operator. This
compilation strategy has the effect of inlining the operator implementations
into a single loop nest. We found that using this strategy instead of a
traditional iterator model approach is critical for getting good performance
from the generated code.

The drawback of push-based query evaluation is that certain operators, such as
deduplication and ordering, must buffer their inputs before processing them.
Rather than implement buffering, we restrict the use of these operators and
replace them with layout-based implementations wherever possible.

For each of the layout operators, we generate code that reads the layout
generated by the serialization pass. This process uses the layout type to
determine where to specialize the layout reading code. Essentially, for each
layout specialization in \Cref{sec:layout-semantics}, there is a corresponding
specialization of the generated code.


\section{Correctness of Filter Elimination}\label{sec:partition-proof}

In this section we discuss the correctness of the filter elimination rule
(\Cref{sec:xform-part}) in detail. We show that the relational semantics is
sufficiently detailed to prove the correctness of the transformation rules.

We say that two programs $q_1$ and $q_2$ are equivalent if they produce the same
value in every context and we denote equivalence as $q_1 \equiv q_2$ according
to the following rule:
\[
  \inference[Equiv]{
    \forall \sigma, \delta, s.~
    \sigma, \delta |- q \evalto s \iff  \sigma, \delta |- q' \evalto s
  }{
    q \equiv q'
  }
\]

Now we prove that the filter elimination rule is semantics-preserving:
\[
  \inference{
    x\ \text{is fresh} & \partition{q}{e}{x} = (q_k, q_v)
  }{
    \textsf{filter}(e = e', q) -> \hidx{}(\as{q_k}{x}, q_v, e')
  }.
\]

\begin{theorem}
  If $\partition{q}{e}{x} = (q_k, q_v)$ and $x$ is a fresh scope,
  then \[\textsf{filter}(e = e', q) \equiv \hidx{}(\as{q_k}{x}, q_v, e').\]
\end{theorem}

\begin{proof}
  By Equiv, the right-hand-side of this implication is equivalent to:
  \begin{align*}
    &\forall \sigma, \delta, s.~\sigma, \delta |-
    \textsf{filter}(e = e', q) \evalto s \iff \\
    &\quad\sigma, \delta |- \hidx{}(\as{q_k}{x}, q_v, e') \evalto s.
  \end{align*}

  By R-HI,
  \begin{align*}
    &\sigma, \delta |- \textsf{filter}(e = e', q) \evalto s \iff \\
    &\quad\sigma, \delta |- \depjoin{}(\as{q_k}{x}, \filter{}(x.e = e', q_v)) \evalto s.
  \end{align*}
           
By the definition of \textsc{partition}, $q_k = \dedup{}(\select{}(\{e\}, q))$
and $q_v = \filter{}(x.e = e, q)$, so
\begin{align*}
  &\sigma, \delta |- \textsf{filter}(e = e', q) \evalto s \iff \\
  &\quad\sigma, \delta |- \depjoin{}(\as{\dedup{}(\select{}(\{e\}, q))}{x},\\
  &\quad\quad\filter{}(x.e = e', \filter{}(x.e = e, q))) \evalto s.
\end{align*}

We can simplify the filter operators to get:
\begin{align*}
  &\sigma, \delta |- \textsf{filter}(e = e', q) \evalto s \iff \\
  &\quad\sigma, \delta |- \depjoin{}(\as{\dedup{}(\select{}(\{e\}, q))}{x},\\
  &\quad\quad\filter{}(x.e = e' \land x.e = e, q)) \evalto s.
\end{align*}
Proving the correctness of this simplification is straightforward and does not
rely on the correctness of the hash-index introduction rule.

By R-Filter and R-Depjoin (and some abuse of notation), this is equivalent to:
\[
  [t ~|~ t <- \textsf{filter}(e = e', q)] = \left[t' ~\Bigg|~
    {\begin{array}{c}
       t <- \dedup{}(\select{}(\{e\}, q)) \\
       t' <- \filter{}(t = e' \land t = e, q)
     \end{array}}
 \right].
\]

At this point there are two cases of interest. In the first case, assume that
$e' \in \dedup{}(\select{}(\{e\}, q))$. By the semantics of \dedup{}, $e'$ will
appear exactly once in this query result if it appears at all. We can conclude
that
\begin{align*}
  &\left[t' ~\Bigg|~
  {\begin{array}{c}
     t <- \dedup{}(\select{}(\{e\}, q)) \\
     t' <- \filter{}(t = e' \land t = e, q)
   \end{array}}
  \right] \\
  &\quad= \begin{array}{c}
       [t' ~|~ t' <- \filter(e' = e' \land e' = e, q)] \concat{} \\
       \quad[t' ~|~ t\neq e', t' <- \filter(e' = t \land t = e, q)]
     \end{array} \\
  &\quad= [t' ~|~ t' <- \filter(e' = e' \land e' = e, q)] \concat{} [~] \\
  &\quad= [t' ~|~ t' <- \filter(e' = e, q)]
\end{align*}

In the second case, assume that $e' \not \in \dedup{}(\select{}(\{e\}, q))$. In
this case, there is no $e$ such that $e = e'$, so $\textsf{filter}(e = e', q) =
[~]$. Similarly, there is no $t$ such that $t = e'$, so \[\left[t' ~\Bigg|~
  {\begin{array}{c}
     t <- \dedup{}(\select{}(\{e\}, q)) \\
     t' <- \filter{}(t = e' \land t = e, q)
   \end{array}}
\right] = [~].\]

In both cases, the two programs are equivalent, so we can conclude that the rule
is semantics-preserving.
\end{proof}

We can conclude from this proof that showing correctness for the transformation
rules is feasible.


\begin{figure}
  {\small
    \begin{alignat*}{3}
      S &::=&&\ \{x_1 \mapsto e_1, \dots, x_m \mapsto e_m\} \quad
      T ::= [C,\ \dots,\ q_n] ~|~ \dots ~|~ [q_1,\ \dots,\ C] \\
      C &::=&&\ [\cdot] ~|~ \textsf{select}(S, C) ~|~ \textsf{filter}(e, C) ~|~
      \textsf{join}(e, q, C) ~|~ \textsf{group-by}(S, E, C) \\
      &\ccol{|}&&\ \textsf{dedup}(C) ~|~ \alist{}(\as{q}{x}, C) ~|~
      \alist{}(\as{C}{x}, q) ~|~ \atuple{\tau}(T) \\
      &\ccol{|}&&\ \hidx{}(\as{C}{x}, q_v, t_k) ~|~ \hidx{}(\as{q_k}{x}, C, t_k) \\ 
      &\ccol{|}&&\ \oidx{}(\as{C}{x}, q_v, t_{lo}, t_{hi}) ~|~ \oidx{}(\as{q_k}{x}, C, t_{lo}, t_{hi})
    \end{alignat*}
  } 
  \caption{The grammar of contexts.}\label{fig:contexts}
\end{figure}

\begin{figure*}
  \begin{small}
  \begin{gather*}
    Id = (Scope?, Name) \quad
    Context = Tuple = \{Id \mapsto Value\} \quad
    Relation = [Tuple] \\
    \sigma : Context \quad
    \delta : Id \mapsto Relation \quad
    s : Id \quad
    t : Tuple \quad
    v : Value \quad
    r : Relation
    \\[2ex]
    \inference[Ctx-Union]{
      \forall n.\ n \not \in \sigma \lor n \not \in \sigma' \\
      \sigma'' = \{n_i \mapsto v_i ~|~ \sigma[n_i] = v_i \lor \sigma'[n_i] =
      v_i\}
    }{
      \sigma \cup \sigma' = \sigma''
    }\quad
    \inference[E-Record]{
      t = \{n_1 \mapsto e_1, \dots, n_m \mapsto e_m\} \\
      \forall i.\ \eto{e_i}{v_i} \\
      t' = \{n_1 \mapsto v_1, \dots, n_m \mapsto v_m\}
    }{
      \eto{t}{t'}
    }\quad\inference[R-Relation]{
      \delta[n] = r
    }{
      \eto{\textsf{relation}(n)}{r}
    }\\[2ex]
    \quad
    \inference[R-Empty]{}{\sigma, \delta |- \emptyset \evalto [\ ]}\quad
    \inference[R-Scalar]{
      \eto{e}{v}
    }{
      \eto{\scalar{n \mapsto e}}{[\{n \mapsto v\}]}
    }\quad\inference[R-Join]{
      \eto{q}{r} & \eto{q'}{r'}  \\
      s = \left[t \cup t' ~|~
        t <- r \quad t' <- r' \quad \eto[\sigma \cup t \cup t', \delta]{e}{\textsf{true}}
      \right]
    }{
      \eto{\textsf{join}(e, q, q')}{r}
    }\\[2ex]
    \inference[R-Filter]{
      \eto{q}{r_q} \\
      r = [t ~|~ t <- r_q \quad \eto[\sigma \cup t, \delta]{e}{\textsf{true}}]
    }{
      \eto{\textsf{filter}(e, q)}{r}
    }\quad
    \inference[R-Select]{
      t\ \text{contains no aggregates} & \eto{q}{r_q} \\
      r = [t'' ~|~ t' <- r_q \quad \eto[\sigma \cup t', \delta]{t}{t''}]
    }{
      \eto{\textsf{select}(t, q)}{r}
    }\\[2ex]
    \inference[R-Dedup]{
      \eto{q}{r_q} & \forall t \in r.\ t \in r_q \\
      \forall t \in r_q.\ \exists i. 1 \leq
      i \leq |r| \land r[i] = t \land \forall j.\ j = i \lor t \neq r[j] 
    }{
      \eto{\textsf{dedup}(q)}{r}
    }\quad\inference[R-DepJoin]{
      \eto{q}{r} &
      r'' = \left[t' ~\Bigg|~
        {\begin{array}{c}
           t <- r \quad t' <- r' \\ t_s = \{(s, f) \mapsto v ~|~
           (f \mapsto v) \in t\} \\ \eto[\sigma \cup t_s, \delta]{q'}{r'}
         \end{array}}
     \right]
   }{
     \eto{\textsf{depjoin}(\as{q}{s}, q')}{r''}
   }\\[2ex]      
   \inference[R-T1]{}{\eto{\atuple{\tau}([~])}{[~]}}\quad
   \inference[R-T2]{
     \tau = \textsf{cross} &
     \eto{q_1}{r_q} \\ \eto{\atuple{\tau}([q_2, \dots, q_n])}{r_{qs}} \\
     r = [t \cup ts ~|~ t <- r_{q}, ts <- r_{qs}]
   }{
     \eto{\atuple{\tau}([q_1, \dots, q_n])}{s}
   }\quad
   \inference[R-T3]{
     \tau = \textsf{concat} \\ \eto{q_1}{r_q} \\ \eto{\atuple{\tau}([q_2, \dots, q_n])}{r_{qs}} 
   }{
     \eto{\atuple{\tau}([q_1, \dots, q_n])}{r_q \concat r_{qs}} 
   }\\[2ex]
   \inference[R-List]{
     \eto{\textsf{depjoin}(\as{q_r}{s}, q)}{r}
   }{
     \eto{\alist{}(\as{q_r}{s}, q)}{r}
   }\enspace
   \inference[R-HI]{
     \eto{\textsf{depjoin}(\as{q_k}{s}, \filter{}(s.x = l, q_v))}{r} \\
     \textsc{schema}(q_k) = [x]
   }{
     \eto{\hidx{}(\as{q_k}{s}, q_v, l)}{r}
   }\enspace
   \inference[R-OI]{
     \eto{\textsf{depjoin}(\as{q_k}{s}, \filter{}(l_{lo}\leq s.x \leq l_{hi},
       q_v))}{r} \\
     \textsc{schema}(q_k) = [x]
   }{
     \eto{\oidx{}(\as{q_k}{s}, q_v, l_{lo}, l_{hi})}{r}
   }
 \end{gather*}
\end{small}
\caption{Execution semantics of the \langname.}\label{fig:runtime-semantics-full}
\end{figure*}


\begin{figure}
{\small
  \begin{align*}
    &\hidx{}(\select{}(\{id_p, id_c\},\\
    &\quad\join{}(enter_p < enter_c \land enter_c < exit_p,\\
    &\quad\quad\select{}(\{id \mapsto id_p, enter \mapsto enter_p, exit \mapsto exit_p\}, log), \\
    &\quad\quad\select{}(\{id \mapsto id_c, enter \mapsto enter_c\}, log)))\ \textsf{as}\ h,\\
    &\quad\alist{}(\select{}(\{enter_p, enter_c\},\\ 
    &\quad\quad\join{}(enter_p < enter_c \land enter_c < exit_p \land id_c = h.id_c \land id_p = h.id_p,\\
    &\quad\quad\quad\select{}(\{id \mapsto id_p, enter \mapsto enter_p, exit \mapsto exit_p\}, log),\\
    &\quad\quad\quad\select{}(\{id \mapsto id_c, enter \mapsto enter_c\}, log)))\ \textsf{as}\ l,\\
    &\quad\quad\ctuple{[\scalar{l.enter_p}, \scalar{l.enter_c}]}),\\
    &\quad(\$pid, \$cid))
  \end{align*}
}\vspace{-3ex}
\caption{Hash-index layout.}\label{fig:hash-code}
\end{figure}

\begin{table*}
  \vspace*{-2ex}
  \begin{threeparttable}
\begin{tabular}{lrrrrrrrrrrr}
\toprule

& \multicolumn{4}{c}{\hyper{}} & \multicolumn{3}{c}{\projname{} (Expert)} & \multicolumn{3}{c}{\projname{} (Optimizer)} \\
\cmidrule(lr){2-5}\cmidrule(lr){6-8}\cmidrule(lr){9-11}
Q\# & Time\tnote{3} & Time\tnote{2} & Mem. & Size & Time\tnote{2} & Mem. & Size & Time\tnote{2} & Mem. & Size \\
\midrule
1  & 19.00 & 0.12 & 15.2 & 17.8 & 0.04 & 2.3 & 0.1 & 0.03 & 2.2 & 0.1 \\
2 \tnote{0} & 5.00 & 10.52 & 238.0 & 206.6 & 0.71 & 10.4 & 41.8 & 0.08 & 2.9 & 102.8 \\
3 \tnote{0} \tnote{1} & 17.00 & 22.30 & 877.1 & 966.8 & 4.73 & 18.3 & 81.0 & 4.98 & 18.4 & 81.0 \\
4  & 8.00 & 0.04 & 17.0 & 17.8 & <0.01 & 2.4 & 0.2 & <0.01 & 2.2 & 0.2 \\
5  & 11.00 & 3.87 & 26.0 & 24.1 & <0.01 & 2.5 & 1.2 & 0.04 & 3.0 & 3.8 \\
6  & 12.00 & 11.75 & 900.5 & 858.8 & 5.69 & 7.2 & 31.0 & --- & --- & --- \\
7  & 12.00 & 0.01 & 16.3 & 17.8 & <0.01 & 2.3 & <0.1 & 0.34 & 2.5 & 0.2 \\
8  & 7.00 & 17.35 & 529.5 & 375.4 & 7.33 & 15.4 & 67.1 & 0.02 & 2.5 & 29.4 \\
9  & 31.00 & 40.60 & 1580.2 & 1550.8 & 54.08 & 358.6 & 365.0 & 241.00 & 317.6 & 323.0 \\
10 \tnote{0} \tnote{1} & 17.00 & <0.01 & 116.8 & 112.2 & 7.82 & 34.5 & 33.1 & 38.10 & 227.2 & 230.6 \\
11  & 8.00 & 16.07 & 89.5 & 68.2 & 0.15 & 2.6 & 5.3 & --- & --- & --- \\
12  & 8.00 & 0.08 & 15.4 & 17.8 & <0.01 & 2.3 & 0.2 & 0.24 & 3.8 & 148.3 \\
14  & 4.00 & <0.01 & 16.0 & 17.8 & <0.01 & 2.2 & <0.1 & <0.01 & 2.3 & <0.1 \\
15  & 13.00 & <0.01 & 14.7 & 17.8 & <0.01 & 2.2 & 0.2 & --- & --- & --- \\
16 \tnote{1} & 38.00 & 3.19 & 32.7 & 35.7 & 1.01 & 5.8 & 3.9 & --- & --- & --- \\
17  & 11.00 & 7.43 & 1238.2 & 1224.7 & 0.04 & 5.0 & 50.8 & <0.01 & 2.4 & 1.6 \\
18 \tnote{0} & 42.00 & 12.47 & 388.7 & 295.7 & 39.18 & 73.9 & 73.5 & --- & --- & --- \\
19  & 28.00 & 0.02 & 19.5 & 17.8 & 0.10 & 2.4 & 0.2 & --- & --- & --- \\
\bottomrule
\end{tabular}
$^0$~Limit clause removed.
$^1$~Run time ordering removed.
$^2$~Specialized.
$^3$~Unspecialized.
\end{threeparttable}

  \caption{Runtime of queries derived from TPC-H (ms). Memory use is the peak
    resident set size during a query (Mb). Size is the layout size (Mb).}\label{tbl:tpch}
  \vspace*{-3ex}
\end{table*}


\end{document}